\documentclass[letterpaper,11pt]{article}
\usepackage{fullpage}
\usepackage{microtype}
\usepackage{times}

\usepackage{amsmath,amsfonts,amssymb}
\usepackage[amsmath,amsthm,thmmarks,hyperref]{ntheorem}
\usepackage{mathtools}
\usepackage{hyperref}
\usepackage[capitalize,nameinlink,noabbrev]{cleveref} %

\usepackage{enumitem}
\usepackage{xcolor}
\usepackage{breakcites}

\usepackage[multiuser,inline,nomargin,marginclue,draft]{fixme}

\usepackage{hdb_macros}

\newcommand{\sm}[1]{\ensuremath{\lambda^{(#1)}}} %

\DeclareMathOperator*{\arcsinh}{arcsinh}
\newcommand{\PPoly}{\cc{P}/\cc{Poly}}

\title{Hardness of Bounded Distance Decoding on Lattices in $\ell_p$
  Norms}

\author{Huck Bennett\thanks{University of Michigan,
    \texttt{hdbco@umich.edu}.}
  \and Chris Peikert\thanks{University of Michigan,
    \texttt{cpeikert@umich.edu}}}

\begin{document}
\maketitle

\begin{abstract}
Bounded Distance Decoding $\BDD_{p,\alpha}$ is the problem of decoding
a lattice when the target point is promised to be within an~$\alpha$
factor of the minimum distance of the lattice, in the~$\ell_{p}$ norm.
We prove that $\BDD_{p, \alpha}$ is $\NP$-hard under randomized
reductions where~$\alpha \to 1/2$ as~$p \to \infty$ (and for
$\alpha=1/2$ when $p=\infty$), thereby showing the hardness of
decoding for distances approaching the unique-decoding radius for
large~$p$. We also show \emph{fine-grained} hardness for
$\BDD_{p,\alpha}$. For example, we prove that for all
$p \in [1,\infty) \setminus 2\Z$ and constants $C > 1, \eps > 0$,
there is no $2^{(1-\eps)n/C}$-time algorithm for $\BDD_{p,\alpha}$ for
some constant~$\alpha$ (which approaches $1/2$ as $p \to \infty$),
assuming the randomized Strong Exponential Time
Hypothesis~(SETH). Moreover, essentially all of our results also hold (under
analogous non-uniform assumptions) for $\BDD$ with
\emph{preprocessing}, in which unbounded precomputation can be applied
to the lattice before the target is available.

Compared to prior work on the hardness of $\BDD_{p,\alpha}$ by Liu,
Lyubashevsky, and Micciancio (APPROX-RANDOM 2008), our results improve
the values of~$\alpha$ for which the problem is known to be $\NP$-hard
for all $p > p_1 \approx 4.2773$, and give the very first fine-grained
hardness for $\BDD$ (in any norm).  Our reductions rely on a special
family of ``locally dense'' lattices in~$\ell_{p}$ norms, which we
construct by modifying the integer-lattice sparsification technique of
Aggarwal and Stephens-Davidowitz (STOC 2018).

\end{abstract}

\section{Introduction}
\label{sec:intro}

Lattices in~$\R^{n}$ are a rich source of computational problems with
applications across computer science, and especially in cryptography
and cryptanalysis. (A lattice is a discrete additive subgroup
of~$\R^{n}$, or equivalently, the set of integer linear combinations
of a set of linearly independent vectors.)  Many
important lattice problems appear intractable, and there is a wealth
of research showing that central problems like the Shortest Vector
Problem~(SVP) and Closest Vector Problem~(CVP) are $\NP$-hard, even to
approximate to within various factors and in various~$\ell_{p}$
norms~\cite{emde81:_anoth_np,journals/jcss/AroraBSS97,conf/stoc/Ajtai98,journals/siamcomp/Micciancio00,journals/tit/Micciancio01,journals/jcss/Khot06,journals/jacm/Khot05,journals/toc/HavivR12,journals/toc/Micciancio12}. (For
the sake of concision, throughout this introduction the term
``$\NP$-hard'' allows for \emph{randomized} reductions, which are
needed in some important cases.)

\paragraph{Bounded Distance Decoding.}

In recent years, the emergence of lattices as a powerful foundation
for cryptography, including for security against quantum attacks, has
increased the importance of other lattice problems. In particular,
many modern lattice-based encryption schemes rely on some form of the
\emph{Bounded Distance Decoding}~(BDD) problem, which is like the
Closest Vector Problem with a promise. An instance of $\BDD_{\alpha}$
for \emph{relative distance} $\alpha > 0$ is a lattice~$\lat$ and a
target point~$\vec{t}$ whose distance from the lattice is guaranteed
to be within an~$\alpha$ factor of the lattice's minimum distance
$\lambda_1(\lat) = \min_{\vec{v} \in \lat \setminus \set{\vec{0}}}
\norm{\vec{v}}$, and the goal is to find a lattice vector within that
distance of~$\vec{t}$; when distances are measured in the~$\ell_{p}$
norm we denote the problem $\BDD_{p,\alpha}$. Note that when
$\alpha < 1/2$ there is a unique solution, but the problem is
interesting and well-defined for larger relative distances as well. 
We also consider \emph{preprocessing} variants of CVP and BDD
(respectively denoted CVPP and BDDP), in which
unbounded precomputation can be applied to the lattice before the
target is available. For example, this can model cryptographic
contexts where a fixed long-term lattice may be shared among many
users.

The importance of BDD(P) to cryptography is especially highlighted by
the Learning With Errors~(LWE) problem of
Regev~\cite{journals/jacm/Regev09}, which is an average-case form of
BDD that has been used (with inverse-polynomial~$\alpha$) in countless
cryptosystems, including several that share a lattice among many users
(see, e.g.,~\cite{DBLP:conf/stoc/GentryPV08}). Moreover, Regev gave a
worst-case to average-case reduction from BDD to LWE, so the security
of cryptosystems is intimately related to the worst-case complexity of
BDD.

Compared to problems like SVP and CVP, the BDD(P) problem has received
much less attention from a complexity-theoretic perspective. We are
aware of essentially only one work showing its $\NP$-hardness: Liu,
Lyubashevsky, and Micciancio~\cite{conf/approx/LiuLM06} proved that
$\BDD_{p,\alpha}$ and even $\BDDP_{p,\alpha}$ are $\NP$-hard for
relative distances approaching $\min \set{1/\sqrt{2}, 1/\sqrt[p]{2}}$,
which is $1/\sqrt{2}$ for $p \geq 2$. A few other works relate BDD(P)
to other lattice problems (in both directions) in regimes where the
problems are not believed to be $\NP$-hard,
e.g.,~\cite{conf/soda/Micciancio08,conf/coco/DadushRS14,conf/icalp/BaiSW16}.
(Dadush, Regev, and Stephens-Davidowitz~\cite{conf/coco/DadushRS14}
also gave a reduction that implies $\NP$-hardness of
$\BDD_{2, \alpha}$ for any $\alpha > 1$, which is larger than the
relative distance of $\alpha = 1/\sqrt{2} + \eps$ achieved
by~\cite{conf/approx/LiuLM06}.)

\paragraph{Fine-grained hardness.}

An important aspect of hard lattice problems, especially for
cryptography, is their \emph{quantitative} hardness. That is, we want
not only that a problem cannot be solved in polynomial time, but that
it cannot be solved in, say, $2^{o(n)}$ time or even $2^{n/C}$ time
for a certain constant~$C$. Statements of this kind can be proven
under generic complexity assumptions like the Exponential Time
Hypothesis~(ETH) of Impagliazzo and
Paturi~\cite{journals/jcss/ImpagliazzoP01} or its variants like Strong
ETH~(SETH), via \emph{fine-grained} reductions that are particularly
efficient in the relevant parameters.

Recently, Bennett, Golovnev, and
Stephens-Davidowitz~\cite{conf/focs/BennettGS17} initiated a study of
the fine-grained hardness of lattice problems, focusing on CVP;
follow-up work extended to SVP and showed more for
CVP(P)~\cite{conf/stoc/AggarwalS18,aggarwal2019finegrained}. The
technical goal of these works is a reduction having good \emph{rank
  efficiency}, i.e., a reduction from $k$-SAT on~$n'$ variables to a
lattice problem in rank~$n=(C+o(1)) n'$ for some constant $C \geq 1$,
which we call the reduction's ``rank inefficiency.''  (All of the
lattice problems in question can be solved in $2^{n+o(n)}$
time~\cite{conf/stoc/AggarwalDRS15,conf/focs/AggarwalDS15,conf/soda/AggarwalS18},
so $C=1$ corresponds to optimal rank efficiency.) We mention that
Regev's BDD-to-LWE reduction~\cite{journals/jacm/Regev09} has optimal
rank efficiency, in that it reduces rank-$n$ BDD to rank-$n$ LWE.
However, to date there are no fine-grained $\NP$-hardness results for
BDD itself; the prior $\NP$-hardness proof for
BDD~\cite{conf/approx/LiuLM06} incurs a large polynomial blowup in
rank.

\subsection{Our Results}
\label{sec:our-results}

We show improved $\NP$-hardness, and entirely new fine-grained
hardness, for Bounded Distance Decoding (and BDD with preprocessing)
in arbitrary~$\ell_{p}$ norms. Our work improves upon the known
hardness of BDD in two respects: the relative distance~$\alpha$, and
the rank inefficiency~$C$ (i.e., fine-grainedness) of the
reductions. As~$p$ grows, both quantities improve, simultaneously
approaching the unique-decoding threshold $\alpha=1/2$ and optimal
rank efficiency of $C=1$ as $p \to \infty$, and achieving those
quantities for $p=\infty$. We emphasize that these are the first
fine-grained hardness results of any kind for BDD, for any~$\ell_{p}$
norm.

Our main theorem summarizing the $\NP$- and fine-grained hardness of
BDD (with and without preprocessing) appears below in
\cref{thm:bdd-hardness}. For $p \in [1,\infty)$ and $C > 1$, the
quantities $\alpha_{p}^{*}$ and $\alpha_{p,C}^{*}$ appearing in the
theorem statement are certain positive real numbers that are
decreasing in~$p$ and~$C$, and approaching~$1/2$ as $p \to \infty$
(for any~$C$).  See \cref{fig:bdd-hardness-bounds} for a plot of their
behavior, \cref{eq:alpha_pC,eq:alpha_p} for their formal definitions,
and \cref{lem:analytic-ub-alpha-p-star} for quite tight
closed-form upper bounds.

\begin{restatable}{theorem}{main}
  \label{thm:bdd-hardness}
  The following hold for $\BDD_{p,\alpha}$ and $\BDDP_{p,\alpha}$ in
  rank~$n$:
  \begin{enumerate}[itemsep=0pt]
  \item For every $p \in [1,\infty)$ and constant
    $\alpha > \alpha_{p}^{*}$ (where
    $\alpha_p^* \leq \tfrac12 \cdot 4.6723^{1/p}$), and for
    $(p,\alpha)=(\infty,1/2)$, there is no polynomial-time algorithm
    for $\BDD_{p,\alpha}$ (respectively, $\BDDP_{p,\alpha}$) unless
    $\NP \subseteq \RP$ (resp.,
    $\NP \subseteq \PPoly$).\label{item:bdd-np-hard}

  \item For every $p \in [1,\infty)$ and constant
    $\alpha > \min\set{\alpha_{p}^{*}, \alpha_{2}^{*}}$, and for
    $(p,\alpha)=(\infty,1/2)$, there is no $2^{o(n)}$-time algorithm
    for $\BDD_{p,\alpha}$ unless randomized ETH
    fails.\label{item:bdd-eth-hardness}
    
  \item For every $p \in [1,\infty) \setminus \set{2}$ and constant
    $\alpha > \alpha_{p}^{*}$, and for $(p,\alpha)=(\infty,1/2)$,
    there is no $2^{o(n)}$-time algorithm for $\BDDP_{p,\alpha}$
    unless non-uniform ETH fails.
    
    Moreover, for every $p \in [1,\infty]$ and
    $\alpha > \alpha_{2}^{*}$ there is no $2^{o(\sqrt{n})}$-time
    algorithm for $\BDDP_{p,\alpha}$ unless non-uniform ETH
    fails.\label{item:bdd-nu-eth-hardness}

  \item For every $p \in [1,\infty) \setminus 2\Z$ and constants
    $C > 1$, $\alpha > \alpha_{p,C}^{*}$, and $\epsilon > 0$, and for
    $(p,C,\alpha)=(\infty,1,1/2)$, there is no
    $2^{n(1-\epsilon)/C}$-time algorithm for $\BDD_{p,\alpha}$
    (respectively, $\BDDP_{p,\alpha}$) unless randomized SETH (resp.,
    non-uniform SETH) fails.\label{item:bdd-seth-hardness}
  \end{enumerate}
\end{restatable}

Although we do not have closed-form expressions for $\alpha_p^*$ and
$\alpha_{p, C}^*$, we do get quite tight closed-form upper bounds (see
\cref{lem:analytic-ub-alpha-p-star}). Moreover, it is easy to
numerically compute close approximations to them, and to the values
of~$p$ at which they cross certain thresholds. For example,
$\alpha_p^* < 1/\sqrt{2}$ for all $p > p_1 \approx 4.2773$, so
\cref{item:bdd-np-hard} of \cref{thm:bdd-hardness} improves on the
prior best relative distance of any $\alpha > 1/\sqrt{2}$ for the
$\NP$-hardness of $\BDD_{p,\alpha}$ in such $\ell_{p}$
norms~\cite{conf/approx/LiuLM06}.

As a few other example values and their consequences under
\cref{thm:bdd-hardness}, we have $\alpha_2^* \approx 1.05006$,
$\alpha_{3,2}^* \approx 1.1418$, and
$\alpha_{3,5}^* \approx 0.917803$. So by~\cref{item:bdd-eth-hardness},
BDD in the Euclidean norm for any relative distance $\alpha > 1.05006$
requires $2^{\Omega(n)}$ time assuming randomized ETH. And
by~\cref{item:bdd-seth-hardness}, for every $\eps > 0$ there is no
$2^{(1-\eps)n/2}$-time algorithm for $\BDD_{3,1.1418}$, and no
$2^{(1-\eps)n/5}$-time algorithm for $\BDD_{3,0.917803}$, assuming
randomized SETH.

\begin{figure}[t]
  \centering 
  \begin{tabular}{cc}
    \includegraphics[scale=0.43]{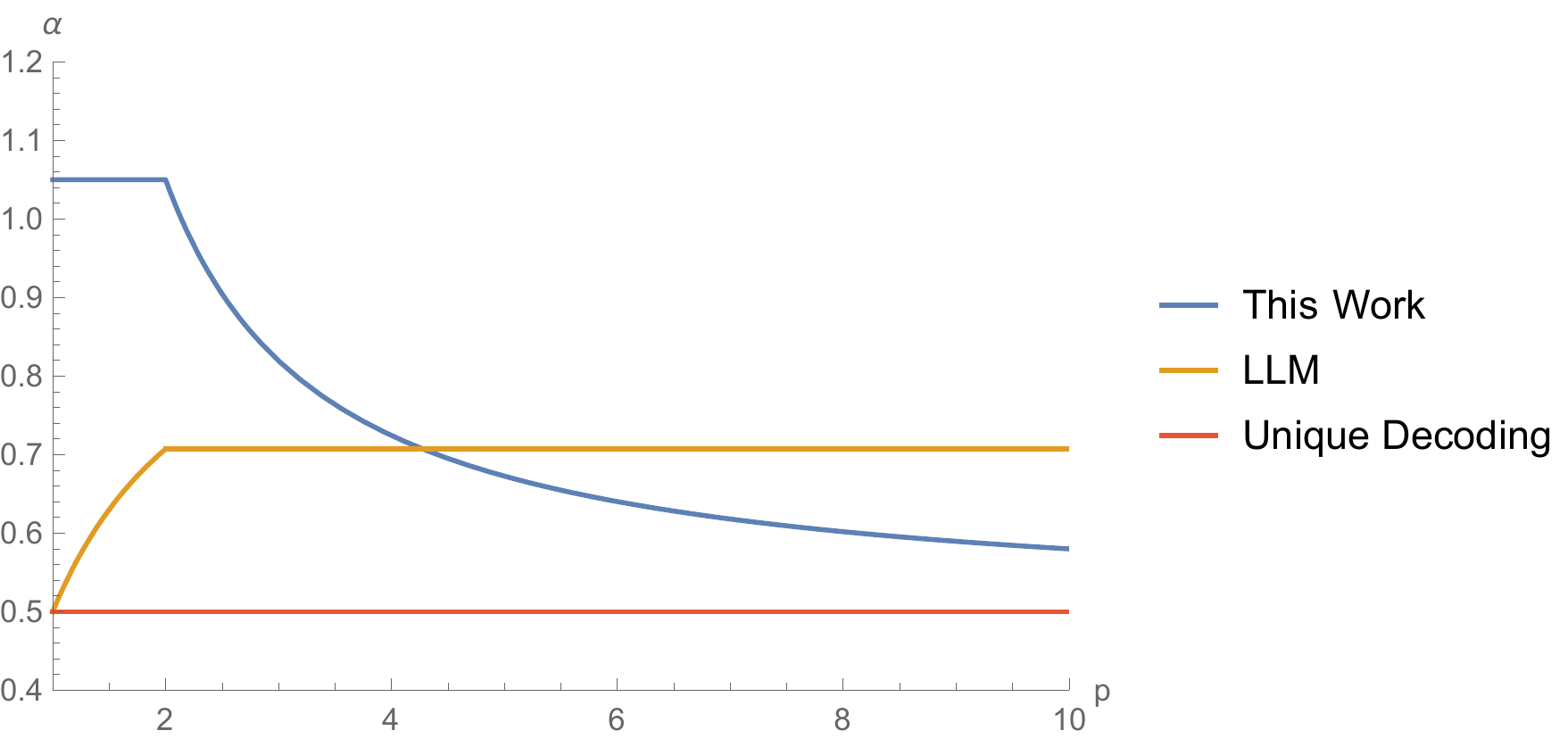} &
    \includegraphics[scale=0.43]{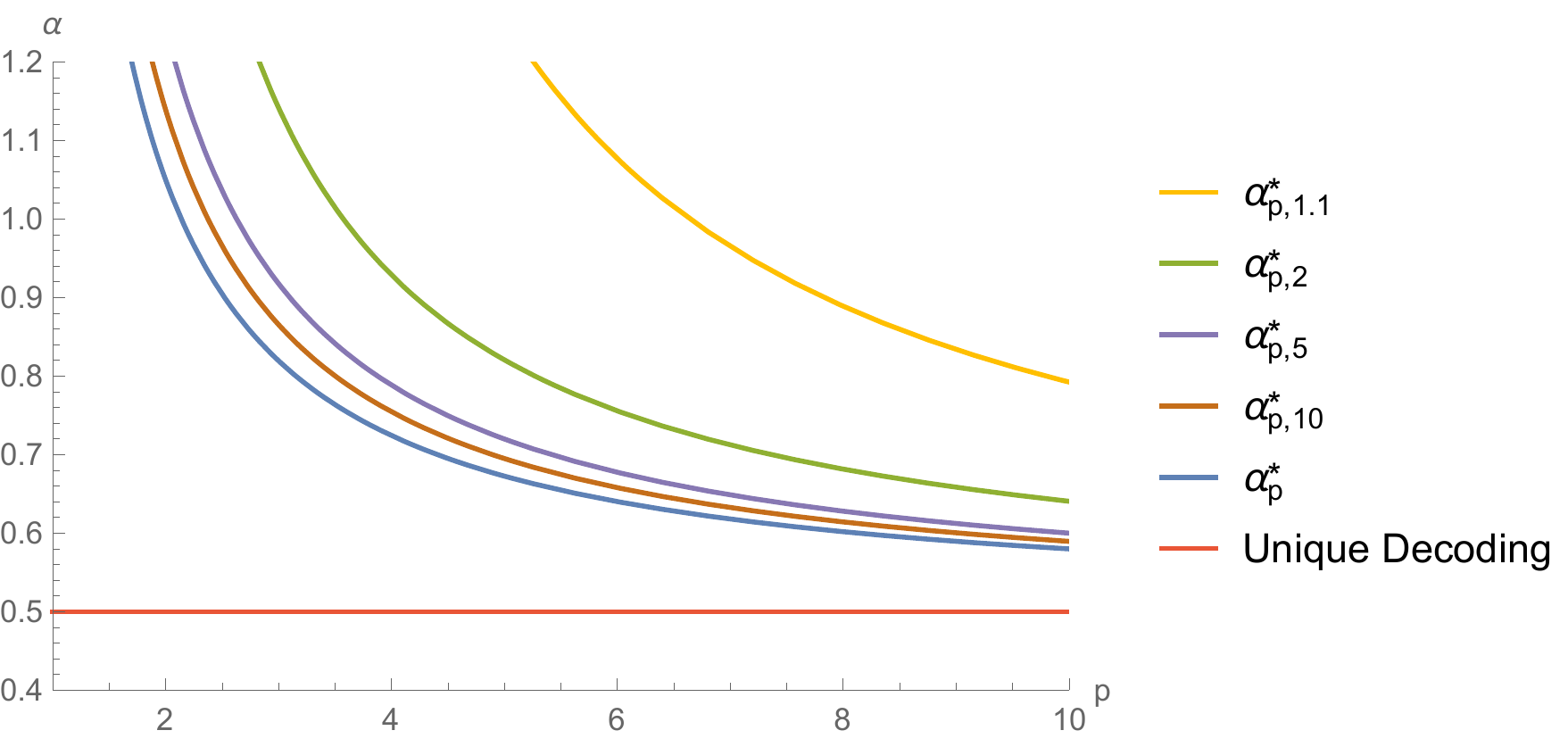}
  \end{tabular}
  \caption{Left: bounds on the relative distances $\alpha = \alpha(p)$
    for which $\BDD_{\alpha, p}$ was proved to be $\NP$-hard in
    the~$\ell_{p}$ norm, in this work and
    in~\cite{conf/approx/LiuLM06}; the crossover point is
    $p_1 \approx 4.2773$.  (The plots include results obtained by norm
    embeddings~\cite{conf/stoc/RegevR06}, hence they are maximized at
    $p=2$.)  Right: our bounds $\alpha_{p,C}^{*}$ on the relative
    distances $\alpha > \alpha_{p,C}^{*}$ for which there is no
    $2^{(1-\eps)n/C}$-time algorithm for $\BDD_{p,\alpha}$ for any
    $\eps > 0$, assuming randomized SETH.}
  \label{fig:bdd-hardness-bounds}
\end{figure}

\subsection{Technical Overview}
\label{sec:technical}

As in prior $\NP$-hardness reductions for SVP and BDD (and
fine-grained hardness proofs for the
former)~\cite{conf/stoc/Ajtai98,journals/siamcomp/Micciancio00,journals/jacm/Khot05,conf/approx/LiuLM06,journals/toc/HavivR12,journals/toc/Micciancio12,conf/stoc/AggarwalS18},
the central component of our reductions is a family of rank-$n$
lattices~$\lat \subset \R^{d}$ and target points~$\vec{t} \in \R^{d}$
having a certain ``local density'' property in a desired~$\ell_{p}$
norm. Informally, this means that~$\lat$ has ``large'' minimum
distance~$\sm{p}_{1}(\lat) := \min_{\vec{v} \in \lat \setminus
  \set{\vec{0}}} \norm{\vec{v}}_{p}$, i.e., there are no ``short''
nonzero vectors, but has many vectors ``close'' to the
target~$\vec{t}$. More precisely, we want $\sm{p}_{1}(\lat) \geq r$
and $N_{p}(\lat, \alpha r, \vec{t}) = \exp(n^{\Omega(1)})$ for some
relative distance~$\alpha$, where
\[ N_{p}(\lat, s, \vec{t}) := \abs{\set{\vec{v} \in \lat :
      \norm{\vec{v}-\vec{t}}_{p} \leq s}} \] denotes the number of
lattice points within distance~$s$ of~$\vec{t}$.

Micciancio~\cite{journals/siamcomp/Micciancio00} constructed locally
dense lattices with relative distance approaching $2^{-1/p}$ in
the~$\ell_{p}$ norm (for every finite $p \geq 1$), and used them to
prove the $\NP$-hardness of $\gamma$-approximate SVP in~$\ell_{p}$ for
any $\gamma < 2^{1/p}$. Subsequently, Liu, Lyubashevsky, and
Micciancio~\cite{conf/approx/LiuLM06} used these lattices to prove the
$\NP$-hardness of BDD in~$\ell_{p}$ for any relative distance
$\alpha > 2^{-1/p}$. However, these works observed that the relative
distance depends on~$p$ in the opposite way from what one might
expect: as~$p$ grows, so does~$\alpha$, hence the associated
$\NP$-hard SVP approximation factors and BDD relative distances
\emph{worsen}. Yet using norm embeddings, it can be shown
that~$\ell_{2}$ is essentially the ``easiest''~$\ell_{p}$ norm for
lattice problems~\cite{conf/stoc/RegevR06}, so hardness in~$\ell_{2}$
implies hardness in~$\ell_{p}$ (up to an arbitrarily small loss in
approximation factor). Therefore, the locally dense lattices
from~\cite{journals/siamcomp/Micciancio00} do not seem to provide any
benefits for $p > 2$ over $p=2$, where the relative distance
approaches $1/\sqrt{2}$.
In addition, the rank of these lattices is a large polynomial in the
relevant parameter, so they are not suitable for proving fine-grained
hardness.\footnote{We mention that Khot~\cite{journals/jacm/Khot05}
  gave a different construction of locally dense lattices with other
  useful properties, but their relative distance is no smaller than
  that of Micciancio's construction in any~$\ell_{p}$ norm, and their
  rank is also a large polynomial in the relevant parameter.}

\paragraph{Local density via sparsification.}

More recently, Aggarwal and
Stephens-Davidowitz~\cite{conf/stoc/AggarwalS18} (building
on~\cite{conf/focs/BennettGS17}) proved fine-grained hardness for
\emph{exact} SVP in~$\ell_{p}$ norms, via locally dense lattices
obtained in a different way. Because they target exact SVP, it
suffices to have local density for relative distance $\alpha=1$, but
for fine-grained hardness they need
$N_{p}(\lat, r, \vec{t}) = 2^{\Omega(n)}$, preferably with a large
hidden constant (which determines the rank efficiency of the
reduction).
Following~\cite{mazo90:_lattice_points,elkies91:_packing_densities},
they start with the integer lattice~$\Z^{n}$ and all-$\tfrac12$s
target vector $\vec{t} = \tfrac12 \vec{1} \in \R^{n}$. Clearly, there
are~$2^{n}$ lattice vectors all at distance $r = \tfrac12 n^{1/p}$
from~$\vec{t}$ in the $\ell_{p}$ norm, but the minimum distance of the
lattice is only~$1$, so the relative distance of the ``close'' vectors
is $\alpha = r$, which is far too large.

To improve the relative distance, they increase the minimum distance
to at least~$r = \tfrac12 n^{1/p}$ using the elegant technique of
\emph{random sparsification}, which is implicit
in~\cite{elkies91:_packing_densities} and was first used for proving
$\NP$-hardness of approximate SVP
in~\cite{journals/jcss/Khot06,journals/jacm/Khot05}. The idea is to
upper-bound the number $N_{p}(\Z^{n}, r, \vec{0})$ of ``short''
lattice points of length at most~$r$, by some~$Q$. Then, by taking a
random sublattice $\lat \subset \Z^{n}$ of determinant (index)
slightly larger than~$Q$, with noticeable probability none of the
``short'' nonzero vectors will be included in~$\lat$, whereas roughly
$2^{n}/Q$ of the vectors ``close'' to~$\vec{t}$ will be in~$\lat$. So,
as long as $Q = 2^{(1-\Omega(1))n}$, there are sufficiently many
lattice vectors at the desired relative distance from~$\vec{t}$.

Bounds for $N_{p}(\Z^{n}, r, \vec{0})$ were given by Mazo and
Odlyzko~\cite{mazo90:_lattice_points}, by a simple but powerful
technique using the theta function
$\Theta_{p}(\tau) := \sum_{z \in \Z} \exp(-\tau \abs{z}^{p})$.  They
showed (see \cref{prop:short-vec-ub}) that
\begin{equation}
  \label{eq:intro-mazo-odlyzko}
 N_{p}(\Z^{n}, r, \vec{0}) \leq \min_{\tau > 0} \exp(\tau \cdot
  r^{p}) \cdot \Theta_{p}(\tau)^{n} = \parens[\Big]{\min_{\tau > 0}
    \exp(\tau/2^{p}) \cdot \Theta_{p}(\tau)}^{n} \ , 
\end{equation}		
where the equality is by $r=\tfrac12 n^{1/p}$. So, Aggarwal and
Stephens-Davidowitz need
$\min_{\tau > 0} \exp(\tau/2^{p}) \cdot \Theta_{p}(\tau) < 2$, and it
turns out that this is the case for every $p > p_{0} \approx
2.1397$. (They also deal with smaller~$p$ by using a different target
point~$\vec{t}$.)

\paragraph{This work: local density for small relative distance.}

For the $\NP$- and fine-grained hardness of BDD we use the same basic
approach as in~\cite{conf/stoc/AggarwalS18}, but with the different
goal of getting local density for as small of a relative distance
$\alpha < 1$ as we can manage. That is, we still have~$2^{n}$ integral
vectors all at distance $r=\tfrac12 n^{1/p}$ from the
target~$\vec{t} = \tfrac12 \vec{1} \in \R^{n}$, but we want to
``sparsify away'' all the nonzero integral vectors of length less
than~$r/\alpha$. So, we want the right-hand side of the Mazo-Odlyzko
bound (\cref{eq:intro-mazo-odlyzko}) to be at most
$2^{(1-\Omega(1))n}$ for as large of a positive hidden constant as we
can manage. More specifically, for any $p \geq 1$ and $C > 1$ (which
ultimately corresponds to the reduction's rank inefficiency) we can
obtain local density of at least $2^{n/C}$ close vectors at any
relative distance greater than
\[ \alpha^{*}_{p,C} := \inf \set{\alpha^{*} > 0 : \min_{\tau > 0}
    \exp(\tau/(2\alpha^{*})^{p}) \cdot \Theta_{p}(\tau) \leq
    2^{1-1/C}} \ . \] The value of~$\alpha^{*}_{p,C}$ is strictly
decreasing in both~$p$ and~$C$, and for large~$C$ and
$p > p_{1} \approx 4.2773$ it drops below the relative distance of
$1/\sqrt{2}$ approached by the local-density construction
of~\cite{journals/siamcomp/Micciancio00} for~$\ell_{2}$ (and
also~$\ell_{p}$ by norm embeddings.) This is the source of our
improved relative distance for the $\NP$-hardness of BDD in
high~$\ell_{p}$ norms.

We also show that obtaining local density by sparsifying the integer
lattice happens to yield a very simple reduction to BDD from the
\emph{exact} version of CVP, which is how we obtain fine-grained
hardness. Given a CVP instance consisting of a lattice and a target
point, we essentially just take their direct sum with the integer
lattice and the $\tfrac12 \vec{1}$ target (respectively), then
sparsify. (See \cref{lem:cvp-to-st-bdd-transform} and
\cref{thm:st-bdd-to-bdd} for details.) Because this results in the
(sparsified) locally dense lattice having $2^{\Omega(n)}$ close
vectors all \emph{exactly} at the threshold of the BDD promise,
concatenating the CVP instance either keeps the target within the
(slightly weaker) BDD promise, or puts it just outside. This is in
contrast to the prior reduction of~\cite{conf/approx/LiuLM06}, where
the close vectors in the locally dense lattices
of~\cite{journals/siamcomp/Micciancio00} are at various distances from
the target, hence a reduction from \emph{approximate}-CVP with a large
constant factor is needed to put the target outside the BDD
promise. While approximating CVP to within any constant factor is
known to be $\NP$-hard~\cite{journals/jcss/AroraBSS97}, no
fine-grained hardness is known for approximate CVP, except for factors
just slightly larger than one~\cite{aggarwal2019finegrained}.

\subsection{Discussion and Future Work}
\label{sec:discussion}

Our work raises a number of interesting issues and directions for
future research. First, it highlights that there are now two
incomparable approaches for obtaining local density in the~$\ell_{p}$
norm---Micciancio's
construction~\cite{journals/siamcomp/Micciancio00}, and sparsifying
the integer
lattice~\cite{elkies91:_packing_densities,conf/stoc/AggarwalS18}---with
each delivering a better relative distance for certain ranges
of~$p$. For $p \in [1,p_{1} \approx 4.2773]$, Micciancio's
construction (with norm embeddings from~$\ell_{2}$, where applicable)
delivers the better relative distance, which approaches
$\min \set{1/\sqrt[p]{2}, 1/\sqrt{2}}$. Moreover, this is essentially optimal in $\ell_2$, where $1/\sqrt{2}$ is
unachievable due to the Rankin bound,
 which says that in~$\R^{n}$ we can have at most $2n$
subunit vectors with pairwise distances of~$\sqrt{2}$ or more.

A first question, therefore, is whether relative distance less than
$1/\sqrt{2}$ can be obtained for all $p > 2$. We conjecture that this
is true, but can only manage to prove it via sparsification for all
$p > p_{1} \approx 4.2773$. More generally, an important open problem
is to give a unified local-density construction that subsumes both of
the above-mentioned approaches in terms of relative distance, and
ideally in rank efficiency as well. In the other direction, another
important goal is to give lower bounds on the relative distance in
general~$\ell_{p}$ norms. Apart from the Rankin bound, the only bound
we are aware of is the trivial one of $\alpha \geq 1/2$ implied by the
triangle inequality, which is essentially tight for~$\ell_{1}$ and tight for~$\ell_{\infty}$
(as shown by~\cite{journals/siamcomp/Micciancio00} and our work,
respectively).

More broadly, for the BDD relative distance parameter~$\alpha$ there
are three regimes of interest: the local-density regime, where we know
how to prove $\NP$-hardness; the unique-decoding regime
$\alpha < 1/2$; and (at least in some~$\ell_{p}$ norms,
including~$\ell_{2}$) the intermediate regime between them. It would
be very interesting, and would seem to require new techniques, to show
$\NP$-hardness outside the local-density regime. One potential route
would be to devise a \emph{gap amplification} technique for BDD,
analogous to how SVP has been proved to be $\NP$-hard to approximate
to within any constant
factor~\cite{journals/jacm/Khot05,journals/toc/HavivR12,journals/toc/Micciancio12}. Gap
amplification may also be interesting in the absence of
$\NP$-hardness, e.g., for the inverse-polynomial relative distances
used in cryptography. Currently, the only efficient gap amplification
we are aware of is a modest one that decreases the relative distance
by any $(1-1/n)^{O(1)}$ factor~\cite{conf/crypto/LyubashevskyM09}.

A final interesting research direction is related to the \emph{unique}
Shortest Vector Problem (uSVP), where the goal is to find a shortest
nonzero vector~$\vec{v}$ in a given lattice, under the promise that it
is unique (up to sign). More generally, approximate uSVP has the
promise that all lattice vectors not parallel to~$\vec{v}$ are a
certain factor~$\gamma$ as long. It is known that \emph{exact} uSVP is
$\NP$-hard in~$\ell_{2}$~\cite{journals/tcs/KumarS01}, and by known
reductions it is straightforward to show the $\NP$-hardness of
\emph{$2$-approximate} uSVP in~$\ell_{\infty}$.  Can recent techniques
help to prove $\NP$-hardness of $\gamma$-approximate uSVP, for some
constant $\gamma > 1$, in~$\ell_{p}$ for some finite~$p$, or
specifically for~$\ell_{2}$? Do $\NP$-hard approximation factors for
uSVP grow smoothly with~$p$?

\paragraph{Acknowledgments.}

We thank Noah Stephens-Davidowitz for sharing his plot-generating code
from~\cite{conf/stoc/AggarwalS18} with us.

\section{Preliminaries}
\label{sec:prelims}

For any positive integer~$q$, we identify the quotient group
$\Z_{q} = \Z/q\Z$ with some set of distinguished representatives,
e.g., $\set{0, 1, \ldots, q-1}$.  Let $B^{+} := (B^t B)^{-1} B^t$
denote the Moore-Penrose pseudoinverse of a real-valued matrix~$B$
with full column rank. Observe that $B^{+} \vec{v}$ is the unique
coefficient vector~$\vec{c}$ with respect to $B$ of any
$\vec{v} = B \vec{c}$ in the column span of~$B$.

\subsection{Problems with Preprocessing}
\label{sec:preprocessing}

In addition to ordinary computational problems, we are also interested
in (promise) problems with \emph{preprocessing}. In such a problem, an
instance $(x_{P}, x_{Q})$ is comprised of a ``preprocessing''
part~$x_{P}$ and a ``query'' part~$x_{Q}$, and an algorithm is allowed
to perform unbounded computation on the preprocessing part
before receiving the query part.

Formally, a preprocessing problem is a relation
$\Pi = \set{((x_{P}, x_{Q}), y)}$ of instance-solution pairs, where
$\Pi_{\text{inst}} := \set{(x_{P}, x_{Q}) : \exists\; y \text{ s.t. }
  ((x_{P}, x_{Q}), y) \in \Pi}$ is the set of problem instances, and
$\Pi_{(x_{P},x_{Q})} := \set{y : ((x_{P}, x_{Q}), y) \in \Pi}$ is the
set of solutions for any particular instance $(x_{P}, x_{Q})$. If
every instance $(x_{P},x_{Q}) \in \Pi_{\text{inst}}$ has exactly one
solution that is either YES or NO, then~$\Pi$ is called a decision
problem.

\begin{definition}
  \label{def:preprocessing-algorithm}
  A \emph{preprocessing algorithm} is a pair $(P,Q)$ where~$P$ is a
  (possibly randomized) function representing potentially unbounded
  computation, and~$Q$ is an algorithm. The execution of $(P,Q)$ on an
  input $(x_{P}, x_{Q})$ proceeds in two phases:
  \begin{itemize}[itemsep=0pt]
  \item first, in the preprocessing phase, $P$ takes~$x_{P}$ as input
    and produces some preprocessed output~$\sigma$;
  \item then, in the query phase, $Q$ takes both~$\sigma$ and~$x_{Q}$ as
    input and produces some ultimate output.
  \end{itemize}
  The running time~$T$ of the algorithm is defined to be the time used
  in the query phase alone, and is considered as a function of the
  total input length $\abs{x_{P}}+\abs{x_{Q}}$. The length of the
  preprocessed output is defined as $A = \abs{\sigma}$, and is also
  considered as a function of the total input length. Note that
  without loss of generality, $A \leq T$.
	
  If $(P,Q)$ is deterministic, we say that it solves preprocessing
  problem~$\Pi$ if $Q(P(x_{P}), x_{Q}) \in \Pi_{(x_{P},x_{Q})}$ for
  all~$(x_{P}, x_{Q}) \in \Pi_{\text{inst}}$. If $(P,Q)$ is
  potentially randomized, we say that it solves~$\Pi$ if
  \[ \Pr[Q(P(x_{P}), x_{Q}) \in \Pi_{(x_{P},x_{Q})}] \geq \frac23 \ \]
  for all $(x_{P},x_{Q}) \in \Pi_{\text{inst}}$, where the probability
  is taken over the random coins of both~$P$ and~$Q$.\footnote{Note
    that it could be the case that some preprocessed outputs fail to
    make the query algorithm output a correct answer on some, or even
    all, query inputs.}
\end{definition}

As shown below using a routine quantifier-swapping argument (as in
Adleman's Theorem~\cite{DBLP:conf/focs/Adleman78}), it turns out that
for $\NP$ relations and decision problems, any randomized
preprocessing algorithm can be derandomized if the length of the query
input~$x_{Q}$ is polynomial in the length of the preprocessing
input~$x_{P}$. So for convenience, in this work we allow for
randomized algorithms, only switching to deterministic ones for our
ultimate hardness theorems.

\begin{lemma}
  \label{lem:weak-strong-equivalent}
  Let preprocessing problem~$\Pi$ be an $\NP$ relation or a decision
  problem for which $\abs{x_{Q}} = \poly(\abs{x_{P}})$ for all
  $(x_{P},x_{Q}) \in \Pi_{\text{inst}}$. If $\Pi$ has a randomized
  $T$-time algorithm, then it has a deterministic
  $T \cdot \poly(\abs{x_{P}}+\abs{x_{Q}})$-time algorithm with
  $T \cdot \poly(\abs{x_{P}}+\abs{x_{Q}})$-length preprocessed output.
\end{lemma}

\begin{proof}
  Let $q(\cdot)$ be a polynomial for which
  $\abs{x_{Q}} \leq q(\abs{x_{P}})$ for all
  $(x_{P},x_{Q}) \in \Pi_{\text{inst}}$. Let $(P,Q)$ be a randomized
  $T$-time algorithm for~$\Pi$, which by standard repetition
  techniques we can assume has probability strictly less than
  $\exp(-q(\abs{x_{P}}))$ of being incorrect on any
  $(x_{P},x_{Q}) \in \Pi_{\text{inst}}$, with only a
  $\poly(\abs{x_{P}}+\abs{x_{Q}})$-factor overhead in the running time
  and preprocessed output length. Fix some arbitrary~$x_{P}$. Then by
  the union bound over all~$(x_{P},x_{Q}) \in \Pi_{\text{inst}}$ and
  the hypothesis, we have
  \[ \Pr[\exists\; (x_{P},x_{Q}) \in \Pi_{\text{inst}} :
    Q(P(x_{P}),x_{Q}) \not\in \Pi_{(x_{P},x_{Q})}] < 1. \] So, there
  exist coins for~$P$ and~$Q$ for which
  $Q(P(x_{P}), x_{Q}) \in \Pi_{(x_{P},x_{Q})}$ for all
  $(x_{P}, x_{Q}) \in \Pi_{\text{inst}}$. By fixing these coins we
  make~$P$ a deterministic function of~$x_{P}$, and we include the
  coins for~$Q$ along with the preprocessed output~$P(x_{P})$, thus
  making~$Q$ deterministic as well. The resulting deterministic
  algorithm solves~$\Pi$ with the claimed resources, as needed.
\end{proof}
	
\paragraph{Reductions for preprocessing problems.}

We need the following notions of reductions for preprocessing
problems. The following generalizes Turing reductions and Cook
reductions (i.e., polynomial-time Turing reductions).

\begin{definition}
  \label{def:preprocessing-cook-reduction}
  A \emph{Turing reduction} from one preprocessing problem~$X$ to
  another one~$Y$ is a pair of algorithms $(R_{P}, R_{Q})$ satisfying
  the following properties: $R_{P}$ is a (potentially randomized)
  function with access to an oracle~$P$, whose output length is
  polynomial in its input length; $R_{Q}$ is an algorithm with access
  to an oracle~$Q$; and if $(P, Q)$ solves problem~$Y$, then
  $(R_{P}^{P}, R_{Q}^{Q})$ solves problem~$X$. Additionally, it is a
  \emph{Cook reduction} if~$R_{Q}$ runs in time polynomial in the
  total input length of~$R_{P}$ and~$R_{Q}$.
\end{definition}
Similarly, the following generalizes mapping reductions and Karp
reductions (i.e., polynomial-time mapping reductions) for decision
problems.

\begin{definition}
  \label{def:preprocessing-mapping-reduction}
  A \emph{mapping reduction} from one preprocessing decision
  problem~$X$ to another one~$Y$ is a pair $(R_{P}, R_{Q})$ satisfying
  the following properties: $R_{P}$ is a deterministic function whose
  output length is polynomial in its input length; $R_{Q}$ is a
  deterministic algorithm; and for any YES (respectively, NO) instance
  $(x_{P}, x_{Q})$ of~$X$, the output pair $(y_{P}, y_{Q})$ is a YES
  (resp., NO) instance of~$Y$, where $(y_{P}, y_{Q})$ are defined as
  follows:
  \begin{itemize}[itemsep=0pt]
  \item first, $R_P$ takes~$x_P$ as input and outputs some
    $(\sigma', y_P)$, where~$\sigma'$ is some ``internal''
    preprocessed output;
  \item then, $R_Q$ takes $(\sigma', x_Q)$ as input and outputs
    some~$y_Q$.
  \end{itemize}
  Additionally, it is a \emph{Karp reduction} if~$R_{Q}$ runs in time
  polynomial in the total input length of~$R_{P}$ and~$R_{Q}$.
\end{definition}

It is straightforward to see that if~$X$ mapping reduces to~$Y$, and
there is a deterministic polynomial-time preprocessing algorithm
$(P_{Y},Q_{Y})$ that solves~$Y$, then there is also one
$(P_{X},Q_{X})$ that solves~$X$, which works as follows:
\begin{enumerate}[itemsep=0pt]
\item the preprocessing algorithm~$P_{X}$, given a preprocessing
  input~$x_{p}$, first computes $(\sigma', y_{P}) = R_{P}(x_{P})$,
  then computes $\sigma_{Y} = P_{Y}(y_{P})$ and outputs
  $\sigma_{X} = (\sigma', \sigma_{Y})$;
\item the query algorithm $Q_{X}$, given
  $\sigma_{X} = (\sigma', \sigma_{Y})$ and a query input~$x_{Q}$,
  computes $y_{Q} = R_{Q}(\sigma', x_{Q})$ and finally outputs
  $Q_{Y}(\sigma_{Y}, y_{Q})$.
\end{enumerate}

\subsection{Lattices}

A \emph{lattice} is the set of all integer linear combinations of some
linearly independent vectors $\vec{b}_1, \ldots, \vec{b}_n$.  It is
convenient to arrange these vectors as the columns of a
matrix. Accordingly, we define a \emph{basis}
$B = (\vec{b}_1, \ldots, \vec{b}_n) \in \R^{d \times n}$ to be a
matrix with linearly independent columns, and the lattice generated by
basis~$B$ as
\[
  \lat(B) := \set[\Big]{\sum_{i=1}^n a_i \vec{b}_i : a_1, \ldots, a_n \in
    \Z} \ .
\]

Let $\mathcal{B}_p^d$ denote the centered unit $\ell_p$ ball in $d$ dimensions.
Given a lattice $\lat \subset \R^d$ of rank $n$, for $1 \leq i \leq n$ let
\[
\sm{p}_i(\lat) := \inf \set{r > 0 : \dim(\lspan(r \cdot \mathcal{B}_p^d \cap \lat)) \geq i}
\]
denote the $i$th successive minimum of $\lat$ with respect to the $\ell_p$ norm.

We denote the $\ell_p$ distance of a vector $\vec{t}$ to a lattice $\lat$ as
\[
\dist_p(\vec{t}, \lat) := \min_{\vec{v} \in \lat} \norm{\vec{v} - \vec{t}}_p \ .
\]

\subsection{Bounded Distance Decoding (with Preprocessing)}

The primary computational problem that we study in this work is the
Bounded Distance Decoding Problem (BDD), which is a version of the
Closest Vector Problem (CVP) in which the target vector is promised to
be relatively close to the lattice.

\begin{definition}
  \label{def:bdd}
  For $1 \leq p \leq \infty$ and $\alpha = \alpha(n) > 0$, the
  \emph{$\alpha$-Bounded Distance Decoding problem in the
    $\ell_p$ norm} ($\BDD_{p,\alpha}$) is the (search) promise problem
  defined as follows. The input is (a basis of) a rank-$n$ lattice
  $\lat$ and a target vector $\vec{t}$ satisfying
  $\dist_p(\vec{t}, \lat) \leq \alpha(n) \cdot \sm{p}_1(\lat)$. The
  goal is to output a lattice vector $\vec{v} \in \lat$ that satisfies
  $\norm{\vec{v} - \vec{t}}_p \leq \alpha(n) \cdot \sm{p}_{1}(\lat)$.

  The preprocessing (search) promise problem $\BDDP_{p,\alpha}$ is
  defined analogously, where the preprocessing input is (a basis of)
  the lattice, and the query input is the target~$\vec{t}$.
\end{definition}

We note that in some works, $\BDD$ is defined to have the goal of
finding a~$\vec{v} \in \lat$ such that
$\norm{\vec{v}-\vec{t}}_{p} = \dist_{p}(\vec{t}, \lat)$. This
formulation is clearly no easier than the one defined above. So, our
hardness theorems, which are proved for the definition above,
immediately apply to the alternative formulation as well.

We also remark that for $\alpha < 1/2$, the promise ensures that there
is a \emph{unique} vector $\vec{v}$ satisfying
$\norm{\vec{v} - \vec{t}}_p \leq \alpha \cdot
\sm{p}_1(\lat)$. However, $\BDD$ is still well defined for
$\alpha \geq 1/2$, i.e., above the unique-decoding radius. As in prior
work, our hardness results for $\BDD_{p,\alpha}$ are limited to this
regime.

To the best of our knowledge, essentially the only previous study of the
$\NP$-hardness of $\BDD$ is due to~\cite{conf/approx/LiuLM06}, which
showed the following
result.\footnote{Additionally,~\cite{conf/coco/DadushRS14} gave a reduction from $\CVP$ to $\BDD_{2, \alpha}$ but only for some $\alpha > 1$. Also,~\cite{conf/stoc/Peikert09,conf/crypto/LyubashevskyM09}
  gave a reduction from $\GapSVP_{\gamma}$ to $\BDD$, but only for
  large $\gamma = \gamma(n)$ for which $\GapSVP$ is not known to be
  $\NP$-hard.}

\begin{theorem}[{{\cite[Corollaries 1 and 2]{conf/approx/LiuLM06}}}]
  \label{thm:llm-bdd-bound}
  For any $p \in [1,\infty)$ and $\alpha > 1/2^{1/p}$, there is no
  polynomial-time algorithm for $\BDD_{p,\alpha}$ (respectively, with
  preprocessing) unless $\NP \subseteq \RP$ (resp., unless
  $\NP \subseteq \PPoly$).
\end{theorem}

Regev and Rosen~\cite{conf/stoc/RegevR06} used norm embeddings to show
that almost any lattice problem is at least as hard in the $\ell_p$
norm, for any $p \in [1,\infty]$, as it is in the~$\ell_2$ norm, up to
an arbitrarily small constant-factor loss in the approximation factor.
In other words, they essentially showed that $\ell_{2}$ is the
``easiest'' $\ell_{p}$ norm for lattice problems. (In addition, their
reduction preserves the rank of the lattice.) Based on
this,~\cite{conf/approx/LiuLM06} observed the following corollary,
which is an improvement on the factor~$\alpha$ from
\cref{thm:llm-bdd-bound} for all $p > 2$.

\begin{theorem}[{{\cite[Corollary 3]{conf/approx/LiuLM06}}}]
  \label{thm:llm-bdd-bound-norm-embeddings}
  For any $p \in [1,\infty)$ and $\alpha > 1/\sqrt{2}$, there is no
  polynomial-time algorithm for $\BDD_{p,\alpha}$ (respectively, with
  preprocessing) unless $\NP \subseteq \RP$ (resp., unless
  $\NP \subseteq \PPoly$).
\end{theorem}

\cref{fig:bdd-hardness-bounds} shows the bounds from
\cref{thm:llm-bdd-bound,thm:llm-bdd-bound-norm-embeddings}
together with the new bounds achieved in this work as a function
of~$p$.

\subsection{Sparsification}
\label{sec:sparsification}

A powerful idea, first used in the context of hardness proofs for
lattice problems in~\cite{journals/jcss/Khot06}, is that of random
lattice sparsification. Given a lattice $\lat$ with basis $B$, we can
construct a random sublattice $\lat' \subseteq \lat$ as
\[
\lat' = \set{\vec{v} \in \lat : \iprod{\vec{z}, B^{+}\vec{v}} = 0 \pmod*{q}}
\]
for uniformly random $\vec{z} \in \Z_q^n$, where $q$ is a suitably
chosen prime.

\begin{lemma}
  \label{lem:sparsification-upper-bound}
  Let $q$ be a prime and let
  $\vec{x}_1, \ldots, \vec{x}_N \in \Z_q^n \setminus
  \set{\vec{0}}$ be arbitrary. Then
  \[
    \Pr_{\vec{z} \gets \Z_{q}^{n}}[\exists\; i \in [N] \text{ such
      that } \iprod{\vec{z}, \vec{x}_i} = 0 \pmod*{q}] \leq \frac{N}{q}
    \ .
  \]
\end{lemma}

\begin{proof}
  We have $\Pr[\iprod{\vec{z}, \vec{x}_i} = 0] = 1/q$ for each
  $\vec{x}_i$, and the claim follows by the union bound.
\end{proof}
The following corollary is immediate.

\begin{corollary}
  \label{cor:sparsification-min-dist-ub}
  Let $q$ be a prime and $\lat$ be a lattice of rank $n$ with
  basis~$B$.  Then for all $r > 0$ and all $p \in [1,\infty]$,
  \[
    \Pr_{\vec{z} \gets \Z_{q}^{n}}[\sm{p}_1(\lat') < r] \leq
    \frac{N^{o}_p(\lat \setminus \set{\vec{0}}, r, \vec{0})}{q} \ ,
  \]
  where
  $\lat' = \set{\vec{v} \in \lat : \iprod{\vec{z}, B^{+}\vec{v}}
    = 0 \pmod*{q}}$.
\end{corollary}

\begin{theorem}[{{\cite[Theorem~3.1]{conf/soda/Stephens-Davidowitz16}}}]
  \label{thm:inhom-sparsification}
  For any lattice $\lat$ of rank~$n$ with basis $B$, prime~$q$, and
  lattice vectors $\vec{x}, \vec{y}_1, \ldots, \vec{y}_N \in \lat$
  such that $B^+ \vec{x} \neq B^+ \vec{y}_i \pmod*{q}$ for all
  $i \in [N]$, we have
  \[
    \frac{1}{q} - \frac{N}{q^2} - \frac{N}{q^{n-1}} \leq \Pr_{\vec{z},
      \vec{c} \gets \Z_{q}^{n}}[\iprod{\vec{z}, B^+\vec{x} + \vec{c}}
    = 0 \pmod*{q} \wedge \iprod{\vec{z}, B^+\vec{y}_i + \vec{c}} \neq
    0 \pmod*{q} \; \forall\; i \in [N]] \leq \frac{1}{q} +
    \frac{1}{q^n} \ .
  \]
\end{theorem}

We will use only the lower bound from \cref{thm:inhom-sparsification},
but we note that the upper bound is relatively tight for $q \gg N$.

\begin{corollary}
  \label{cor:sparsification-dist-lb}
  For any $p \in [1,\infty]$ and $r \geq 0$, lattice $\lat$ of
  rank~$n$ with basis~$B$, vector~$\vec{t}$, prime~$q$, and lattice
  vectors $\vec{v}_{1}, \ldots, \vec{v}_{N} \in \lat$ such that
  $\norm{\vec{v}_{i} - \vec{t}}_{p} \leq r$ for all~$i \in [N]$ and
  such that all the $B^{+} \vec{v}_{i} \bmod q$ are distinct, we have
  \[
    \Pr_{\vec{z}, \vec{c} \gets \Z_{q}^{n}}[\dist_p(\vec{t} +
    B\vec{c}, \lat') \leq r]
    \geq \frac{N}{q} - \frac{N(N-1)}{q^2} - \frac{N(N-1)}{q^{n-1}} \ ,
  \]
  where
  $\lat' = \set{\vec{v} \in \lat : \iprod{\vec{z}, B^{+}\vec{v}} = 0
    \pmod*{q}}$.
\end{corollary}

\begin{proof}
  Observe that for each $i \in [N]$, the events
  \[
    E_{i} := [\iprod{\vec{z}, B^+\vec{v}_i} = 0 \pmod*{q} \text{ and }
    \iprod{\vec{z}, B^+\vec{v}_j} \neq 0 \pmod*{q} \text{ for all
      $j \neq i$}]
  \]
  are disjoint, and by invoking \cref{thm:inhom-sparsification} with
  $\vec{x} = \vec{v}_{i}$ and the $\vec{y}_{j}$ being the remaining
  $\vec{v}_{k}$ for $k \neq i$, we have
  \[ \Pr_{\vec{z}, \vec{c}}[E_{i}] \geq \frac{1}{q} - \frac{N -
      1}{q^{2}} - \frac{N - 1}{q^{n-1}} \ .
  \] 
  Also observe that if~$E_{i}$ occurs, then
  $\vec{v}_{i} + B \vec{c} \in \lat'$ (also
  $\vec{v}_{j} + B \vec{c} \not\in \lat'$ for all $j \neq i$, but we
  will not need this). Therefore,
  \[ \dist_{p}(\vec{t} + B \vec{c}, \lat') \leq \norm{\vec{t} + B
      \vec{c} - (\vec{v}_{i} + B \vec{c})} = \norm{\vec{t} -
      \vec{v}_{i}} \leq r \ . \] So, the probability in the left-hand
  side of the claim is at least
  \[ \Pr_{\vec{z}, \vec{c}}\bracks[\Big]{\bigcup_{i \in [N]} E_{i}} =
    \sum_{i \in [N]} \Pr_{\vec{z}, \vec{c}}[E_{i}] \geq \frac{N}{q} -
    \frac{N(N-1)}{q^{2}} - \frac{N(N-1)}{q^{n-1}} \ .\]
\end{proof}

\subsection{Counting Lattice Points in a Ball}

Following~\cite{conf/stoc/AggarwalS18}, for any discrete set~$A$ of
points (e.g., a lattice, or a subset thereof), we denote the number of
points in~$A$ contained in the closed and open (respectively)~$\ell_p$
ball of radius~$r$ centered at a point~$\vec{t}$ as
\begin{align}
  N_p(A, r, \vec{t}) &:= \card{\set{\vec{y} \in A : \norm{\vec{y} -
  \vec{t}}_p \leq r}} \ , \\
  N^{o}_p(A, r, \vec{t}) &:= \card{\set{\vec{y} \in A : \norm{\vec{y} -
  \vec{t}}_p < r}} \ .
\end{align}
Clearly, $N^{o}_{p}(A,r,\vec{t}) \leq N_{p}(A,r,\vec{t})$.

For $1 \leq p < \infty$ and $\tau > 0$ define
\[
  \Theta_p(\tau) := \sum_{z \in \Z} \exp(-\tau \abs{z}^p) \ .
\]
We use the following upper bound due to Mazo and
Odlyzko~\cite{mazo90:_lattice_points} on the number of short vectors
in the integer lattice.  We include its short proof for completeness.

\begin{proposition}[{\cite{mazo90:_lattice_points}}]
  \label{prop:short-vec-ub}
  For any $p \in [1,\infty)$, $r > 0$, and $n \in \N$,
  \[
    N_p(\Z^n, r, \vec{0}) \leq \min_{\tau > 0} \exp(\tau r^p) \cdot
    \Theta_p(\tau)^n \ .
  \]
\end{proposition}

\begin{proof}
  For $\tau > 0$ we have
  \[
    \Theta_p(\tau)^n = \sum_{\vec{z} \in \Z^n} \exp(-\tau
    \norm{\vec{z}}_p^p) \geq \sum_{\vec{z} \in \Z^n \cap
      r\mathcal{B}_{p}^{n}} \exp(-\tau \norm{\vec{z}}_p^p) \geq
    \exp(-\tau r^p) \cdot N_p(\Z^n, r, \vec{0}) \ .
  \]
  The result follows by rearranging and taking the minimum over all
  $\tau > 0$.
\end{proof}

\subsection{Hardness Assumptions}
\label{sec:hardness-assumptions}

We recall the Exponential Time Hypothesis (ETH) of Impagliazzo and
Paturi~\cite{journals/jcss/ImpagliazzoP01}, and several of its
variants. These hypotheses make stronger assumptions about the
complexity of the $k$-SAT problem than the assumption $\P \neq \NP$,
and serve as highly useful tools for studying the fine-grained
complexity of hard computational problems. Indeed, we will show that
strong fine-grained hardness for $\BDD$ follows from these hypotheses.

\begin{definition}
  \label{def:eth}
  The \emph{(randomized) Exponential Time Hypothesis} ((randomized)
  ETH) asserts that there is no (randomized) $2^{o(n)}$-time algorithm
  for $3$-SAT on~$n$ variables.
\end{definition}

\begin{definition}
  \label{def:seth}
  The \emph{(randomized) Strong Exponential Time Hypothesis}
  ((randomized) SETH) asserts that for every $\eps > 0$ there exists
  $k \in \Z^+$ such that there is no (randomized) $2^{(1-\eps)n}$-time
  algorithm for $k$-SAT on~$n$ variables.
\end{definition}

For proving hardness of lattice problem with preprocessing, we define
(Max-)$k$-SAT with preprocessing as follows. The preprocessing input
is a size parameter~$n$, encoded in unary. The query input is a
$k$-SAT formula $\phi$ with~$n$ variables and~$m$ (distinct) clauses,
together with a threshold $W \in \set{0, \ldots m}$ in the case of
Max-$k$-SAT. For $k$-SAT, it is a YES instance if~$\phi$ is
satisfiable, and is a NO instance otherwise. For Max-$k$-SAT, it is a
YES instance if there exists an assignment to the variables of~$\phi$
that satisfies at least~$W$ of its clauses, and is a NO instance
otherwise.

Observe that because the preprocessing input is just~$n$, a
preprocessing algorithm for (Max-)$k$-SAT with preprocessing is
equivalent to a (non-uniform) family of circuits for the problem
\emph{without} preprocessing. Also, for any fixed~$k$, because there
are only $O(n^{k})$ possible clauses on~$n$ variables, the length of
the query input for (Max-)$k$-SAT instances having preprocessing
input~$n$ is $\poly(n)$, so we get the following corollary
of~\cref{lem:weak-strong-equivalent}.

\begin{corollary}
  \label{cor:preprocessing-k-sat-rerandomization}
  If (Max-)$k$-SAT with preprocessing has a randomized $T(n)$-time
  algorithm, then it has a deterministic $T(n) \cdot \poly(n)$-time
  algorithm using $T(n) \cdot \poly(n)$-length preprocessed output.
\end{corollary}

Following,
e.g.,~\cite{conf/focs/Stephens-Davidowitz19,aggarwal2019finegrained},
we also define \emph{non-uniform} variants of ETH and SETH, which deal
with the complexity of $k$-SAT with preprocessing. More precisely,
non-uniform ETH asserts that no family of size-$2^{o(n)}$ circuits
solves $3$-SAT on~$n$ variables (equivalently, $3$-SAT with
preprocessing does not have a $2^{o(n)}$-time algorithm), and
non-uniform SETH asserts that for every $\eps > 0$ there exists
$k \in \Z^+$ such that no family of circuits of size $2^{(1 - \eps)n}$
solves $k$-SAT on~$n$ variables (equivalently, $k$-SAT with
preprocessing does not have a $2^{(1-\eps)n}$-time algorithm). These
hypotheses are useful for analyzing the fine-grained complexity of
preprocessing problems.

One might additionally consider ``randomized non-uniform'' versions of
(S)ETH. However, \cref{cor:preprocessing-k-sat-rerandomization} says
that a randomized algorithm for (Max-)$k$-SAT with preprocessing can
be derandomized with only polynomial overhead, so randomized
non-uniform (S)ETH is equivalent to (deterministic) non-uniform
(S)ETH, so we only consider the latter.

Finally, we remark that one can define weaker versions of randomized
or non-uniform (S)ETH with Max-$3$-SAT (respectively, Max-$k$-SAT) in
place of $3$-SAT (resp., $k$-SAT). Many of our results hold even under
these weaker hypotheses. 
In particular, the derandomization result in~\cref{cor:preprocessing-k-sat-rerandomization}
applies to both $k$-SAT and Max-$k$-SAT.

\section{\texorpdfstring{Hardness of $\BDD_{p,\alpha}$}{Hardness of BDD\_{p,alpha}}}
\label{sec:bdd}

In this section, we present our main result by giving a reduction from
a known-hard variant $\GapCVP'_p$ of the Closest Vector Problem (CVP)
to $\BDD$.  We peform this reduction in two main steps.
\begin{enumerate}
\item First, in~\cref{sec:stbdd-to-bdd} we define a variant of
  $\BDD_{p, \alpha}$, which we call $(S, T)$-$\BDD_{p,
    \alpha}$. Essentially, an instance of this problem is a lattice
  that may have up to~$S$ ``short'' nonzero vectors of~$\ell_{p}$ norm
  bounded by some~$r$, and a target vector that is ``close''
  to---i.e., within distance~$\alpha r$ of---at least~$T$ lattice
  vectors. (The presence of short vectors prevents this from being a
  true $\BDD_{p,\alpha}$ instance.) We then give a reduction, for
  $S \ll T$, from $(S, T)$-$\BDD_{p, \alpha}$ to $\BDD_{p, \alpha}$
  itself, using sparsification.
\item Then, in~\cref{sec:cvp-to-stbdd} we reduce from $\GapCVP'_p$ to
  $(S,T)$-$\BDD_{p,\alpha}$ for suitable $S \ll T$ whenever~$\alpha$
  is sufficiently large as a function of~$p$ (and the desired rank
  efficiency), based on analysis given in \cref{sec:setting-params}
  and \cref{lem:analytic-ub-alpha-p-star}.
\end{enumerate}

\subsection{$(S,T)$-BDD to BDD}
\label{sec:stbdd-to-bdd}

We start by defining a special decision variant of BDD. Essentially,
the input is a lattice and a target vector, and the problem is to
distinguish between the case where there are few ``short'' lattice
vectors but many lattice vectors ``close'' to the target, and the case
where the target is not close to the lattice. There is a gap factor
between the ``close'' and ``short'' distances, and for technical
reasons we count only those ``close'' vectors having binary
coefficients with respect to the given input basis.

\begin{definition}
  \label{def:st-bdd}
  Let $S = S(n), T = T(n) \geq 0$, $p \in [1,\infty]$, and
  $\alpha = \alpha(n) > 0$. An instance of the decision promise
  problem $(S, T)$-$\BDD_{p,\alpha}$ is a lattice
  basis~$B \in \R^{d \times n}$, a distance $r > 0$, and a target
  $\vec{t} \in \R^d$.
  \begin{itemize}
  \item It is a YES instance if
    $N^{o}_p(\lat(B) \setminus \set{\vec{0}}, r, \vec{0}) \leq S(n)$ and
    $N_{p}(B \cdot \bit^{n}, \alpha r, \vec{t}) \geq T(n)$.
  \item It is a NO instance if $\dist_p(\vec{t}, \lat(B)) > \alpha r$.
  \end{itemize}
  The search version is: given a YES instance $(B,r,\vec{t})$, find a
  $\vec{v} \in \lat(B)$ such that
  $\norm{\vec{v} - \vec{t}}_{p} \leq \alpha r$.

  The preprocessing search and decision problems
  $(S,T)$-$\BDDP_{p,\alpha}$ are defined analogously, where the
  preprocessing input is~$B$ and~$r$, and the query input
  is~$\vec{t}$.
\end{definition}

We stress that in the preprocessing problems $\BDDP$, the distance~$r$
is part of the preprocessing input; this makes the problem no harder
than a variant where~$r$ is part of the query input. So, our hardness
results for the above definition immediately apply to that variant as
well. However, our reduction \emph{from} $(S,T)$-$\BDDP$ (given in
\cref{lem:0t-bdd}) critically relies on the fact that~$r$ is part of
the preprocessing input.

Clearly, there is a trivial reduction from the decision version of
$(S,T)$-$\BDD_{p,\alpha}$ to its search version (and similarly for the
preprocessing problems): just call the oracle for the search problem
and test whether it returns a lattice vector within
distance~$\alpha r$ of the target. So, to obtain more general results,
our reductions involving $(S,T)$-$\BDD$ will be
\emph{from} the search version, and \emph{to} the decision version.

\paragraph{Reducing to BDD.}

We next observe that for $S(n)=0$ and any $T(n) > 0$, there is
\emph{almost} a trivial reduction from $(S,T)$-$\BDD_{p,\alpha}$ to
ordinary $\BDD_{p,\alpha}$, because YES instances of the former
satisfy the $\BDD_{p,\alpha}$ promise. (See below for the easy proof.)
The only subtlety is that we want the $\BDD_{p,\alpha}$ oracle to
return a lattice vector that is within distance $\alpha r$ of the
target; recall that the definition of $\BDD_{p,\alpha}$ only
guarantees distance $\alpha \cdot \sm{p}_{1}(\lat(B))$. This issue is
easily resolved by modifying the lattice to \emph{upper bound} its minimum
distance by~$r$, which increases the lattice's rank by one. (For the
alternative definition of $\BDD$ described after \cref{def:bdd}, the
trivial reduction works, and no increase in the rank is needed.)

\begin{lemma}
  \label{lem:0t-bdd}
  For any $T(n) > 0$, $p \in [1,\infty]$, and
  $\alpha = \alpha(n) > 0$, there is a deterministic Cook reduction
  from the search version of $(0, T(n))$-$\BDD_{p,\alpha}$ (resp.,
  with preprocessing) in rank~$n$ to $\BDD_{p,\alpha}$ (resp., with
  preprocessing) in rank $n+1$.
\end{lemma}

\begin{proof}
  The reduction works as follows. On input $(B, r, \vec{t})$, call the
  $\BDD_{p,\alpha}$ oracle on
  \[ B' :=
    \begin{pmatrix}
      B & 0 \\ 0 & r
    \end{pmatrix}, \quad
    \vec{t}' :=
    \begin{pmatrix}
      \vec{t} \\ 0
    \end{pmatrix} \ ,
  \]
  and (without loss of generality) receive from the oracle a vector
  $\vec{v}' = (\vec{v}, zr)$ for some $\vec{v} \in \lat$ and
  $z \in \Z$. Output $\vec{v}$.

  We analyze the reduction. Let $\lat = \lat(B)$ and
  $\lat' = \lat(B')$. Because the input is a YES instance, we have
  $N^{o}_{p}(\lat \setminus \set{\vec{0}}, r, \vec{0}) = 0$ and hence
  $\sm{p}_{1}(\lat) \geq r$, so $\sm{p}_{1}(\lat') = r$. Moreover,
  $N_{p}(B \cdot \bit^{n}, \alpha r, \vec{t}) > 0$ implies that
  $\dist_{p}(\vec{t}', \lat') = \dist(\vec{t}, \lat) \leq \alpha r =
  \alpha \cdot \sm{p}_{1}(\lat')$. So, $(B',\vec{t}')$ satisfies the
  $\BDD_{p,\alpha}$ promise, hence the oracle is obligated to return
  some $\vec{v}' = (\vec{v}, zr) \in \lat'$ where $\vec{v} \in \lat$
  and
  $\alpha r = \alpha \sm{p}_1(\lat') \geq \norm{\vec{v}' - \vec{t}'}_{p} \geq \norm{\vec{v} -
    \vec{t}}_{p}$. Therefore, the output~$\vec{v}$ of the reduction is
  a valid solution.

  Finally, observe that all of the above also constitutes a valid
  reduction for the preprocessing problems, because~$B'$ depends only
  on the preprocessing part $B,r$ of the input.
\end{proof}

We now present a more general randomized reduction from $(S,T)$-$\BDD_{p,\alpha}$
to $\BDD_{p,\alpha}$, which works whenever $T(n) \geq 10 S(n)$. The
essential idea is to sparsify the input lattice, so that with some
noticeable probability no short vectors remain, but at least one
vector close to the target does remain. In this case, the result will be an
instance of $(0,1)$-$\BDD_{p,\alpha}$, which reduces to
$\BDD_{p,\alpha}$ as shown above.

We note that the triangle inequality precludes the existence of
$(S, T)$-$\BDD_{p,\alpha}$ instances with $T > S+1$ and
$\alpha \leq 1/2$, so with this approach we can only hope to show
hardness of $\BDD_{p,\alpha}$ for $\alpha > 1/2$, i.e., the
unique-decoding regime remains out of reach.

\begin{theorem}
  \label{thm:st-bdd-to-bdd}
  For any $S = S(n) \geq 1$ and $T = T(n) \geq 10 S$ that is
  efficiently computable (for unary~$n$), $p \in [1,\infty]$, and
  $\alpha = \alpha(n) > 0$, there is a randomized Cook reduction with
  no false positives from the search version of
  $(S, T)$-$\BDD_{p,\alpha}$ (resp., with preprocessing) in rank~$n$
  to $\BDD_{p,\alpha}$ (resp., with preprocessing) in rank $n+1$.
\end{theorem}

\begin{proof}
  By \cref{lem:0t-bdd}, it suffices to give such a reduction to
  $(0,1)$-$\BDD_{p,\alpha}$ in rank~$n$, which works as follows. On
  input $(B, r, \vec{t})$, let $\lat = \lat(B)$.  First, randomly
  choose a prime~$q$ where $10T \leq q \leq 20T$. Then sample
  $\vec{z}, \vec{c} \in \Z_q^n$ independently and uniformly at random,
  and define
  \[
    \lat' := \set{\vec{v} \in \lat : \iprod{\vec{z}, B^{+}\vec{v}} = 0
      \pmod*{q}} \text{ and } \vec{t}' := \vec{t} + B\vec{c} \ .
  \]
  Let $B'$ be a basis of $\lat'$. (Such a basis is efficiently computable from
  $B$, $\vec{z}$, and $q$. See,
  e.g.,~\cite[Claim~2.15]{conf/soda/Stephens-Davidowitz16}.) Invoke
  the $(0,1)$-$\BDD_{p,\alpha}$ oracle on $(B', r, \vec{t}')$, and
  output whatever the oracle outputs.

  We now analyze the reduction. We are promised that $(B,r,\vec{t})$
  is a YES instance of $(S,T)$-$\BDD_{p,\alpha}$, and it suffices to
  show that $(B', r, \vec{t}')$ is a YES instance of
  $(0,1)$-$\BDD_{p,\alpha}$, i.e., $\sm{p}_{1}(\lat') \geq r$ and
  $\dist_{p}(\vec{t}', \lat') \leq \alpha r$, with some positive
  constant probability. By \cref{cor:sparsification-min-dist-ub} we
  have
  \[
    \Pr[\sm{p}_1(\lat') < r] \leq \frac{N^{o}_{p}(\lat \setminus
      \set{\vec{0}}, r, \vec{0})}{q} \leq \frac{S}{q} \leq
    \frac{1}{100} \ .
  \]
  Furthermore, because there are~$T$ vectors $\vec{v}_{i} \in \lat$
  for which $\norm{\vec{v}_{i} - \vec{t}}_{p} \leq \alpha r$, and
  their coefficient vectors $B^{+} \vec{v}_{i} \in \bit^{n}$ are
  distinct (as integer vectors, and hence also modulo~$q$), by
  \cref{cor:sparsification-dist-lb} we have
  \[
    \Pr[\dist_p(\vec{t}', \lat') \leq \alpha r] \geq \frac{T}{q} -
    \frac{T^{2}}{q^2} - \frac{T^{2}}{q^{n-1}} \geq \frac{1}{20} -
    \frac{1}{400} - \frac{1}{400 q^{n-3}} \ .
  \]
  Therefore, by the union bound we have
  \begin{equation*}
    \Pr[\sm{p}_1(\lat') \geq r \text{ and } \dist_p(\vec{t}', \lat') \leq \alpha r] 
    \geq \frac{1}{20} - \frac{1}{400} - \frac{1}{400q^{n-3}} - \frac{1}{100}
    \geq \frac{1}{40}
  \end{equation*} 
  for all $n \geq 3$, as desired.

  Finally, the above also constitutes a valid reduction for the
  preprocessing problems (in the sense of
  \cref{def:preprocessing-cook-reduction}), because~$B'$ depends only
  on~$B$ from the preprocessing part of the input and the reduction's
  own random choices (and~$r$ remains unchanged).
\end{proof}

\subsection{GapCVP' to $(S,T)$-BDD}
\label{sec:cvp-to-stbdd}

Here we show that a known-hard variant of the (exact) Closest Vector
Problem reduces to $(S,T)$-BDD (in its decision version).

\begin{definition}
  \label{def:gap-cvp-prime}
  For $p \in [1,\infty]$, the (decision) promise problem
  $\GapCVP'_{p}$ is defined as follows: an instance consists of
  a basis $B \in \R^{d \times n}$ and a target vector
  $\vec{t} \in \R^d$.
  \begin{itemize}
  \item It is a YES instance if there exists $\vec{x} \in \bit^n$ such
    that $\norm{B\vec{x} - \vec{t}}_p \leq 1$.
  \item It is a NO instance if $\dist_p(\vec{t}, \lat(B)) > 1$.
  \end{itemize}
  The preprocessing (decision) promise problem $\GapCVPP'_{p}$
  is defined analogously, where the preprocessing input is~$B$ and the
  query input is~$\vec{t}$.
\end{definition}

Observe that for $\GapCVP'_{p}$ the distance threshold is~$1$ (and not
some instance-dependent value) without loss of generality, because we
can scale the lattice and target vector. The same goes for
$\GapCVPP'_{p}$, with the caveat that any instance-dependent distance
threshold would need to be included in the \emph{preprocessing} part
of the input, not the query part. (See \cref{rem:gapcvpp-variants}
below for why this is essentially without loss of generality, under a
mild assumption on the $\GapCVPP'_{p}$ instances.) We remark that some
works define these problems with a stronger requirement that in the NO
case, $\dist_p(z \vec{t}, \lat(B)) > r$ for all
$z \in \Z \setminus \set{0}$. We will not need this stronger
requirement, and some of the hardness results for $\GapCVP'$ that we
rely on are not known to hold with it, so we use the weaker
requirement.

We next describe a simple transformation on lattices and target
vectors: we essentially take a direct sum of the input lattice with
the integer lattice of any desired dimension~$n$ and append an
all-$\frac12$s vector to the target vector.

\begin{lemma}
  \label{lem:cvp-to-st-bdd-transform}
  For any $n' \leq n$, define the following transformations that map a
  basis~$B'$ of a rank-$n'$ lattice~$\lat'$ to a basis~$B$ of a
  rank-$n$ lattice~$\lat$, and a target vector~$\vec{t}'$ to a target
  vector~$\vec{t}$:
  \begin{equation}
    \label{eq:Bt}
    B := \begin{pmatrix}
      \tfrac12 B' & 0 \\
      I_{n'} & 0 \\
      0 & I_{n-n'}
    \end{pmatrix}, \qquad \vec{t} := \frac{1}{2}
    \begin{pmatrix}
      \vec{t}' \\
      \vec{1}_{n'} \\
      \vec{1}_{n-n'}
    \end{pmatrix} ,
  \end{equation}
  and define
  \begin{equation}
    \label{eq:s_p}
    s_{p} = s_{p}(n) := \tfrac12 (n+1)^{1/p} \text{ for } p \in [1,\infty),
    \text{ and } s_{\infty} := 1/2.
  \end{equation}
   Then:
  \begin{enumerate}[itemsep=0pt]
  \item
    $N^{o}_{p}(\lat, r, \vec{0}) \leq N^{o}_{p}(\Z^{n}, r,
    \vec{0})$ for all $r \geq 0$;\label{item:num-short}
  \item if there exists an $\vec{x} \in \bit^{n'}$ such that
    $\norm{B' \vec{x} - \vec{t}'}_{p} \leq 1$, then
    $N_{p}(B \cdot \bit^{n}, s_{p}, \vec{t}) \geq
    2^{n-n'}$;\label{item:yes-num-close}
  \item if $\dist_{p}(\vec{t}', \lat') > 1$ then
    $\dist_{p}(\vec{t},\lat) > s_{p}$.\label{item:no-far}
  \end{enumerate}
\end{lemma}

\begin{proof}
  \cref{item:num-short} follows immediately by construction of~$B$,
  because vectors
  $\vec{v}' = (\tfrac12 B' \vec{x}, \vec{x}, \vec{y}) \in \lat$ for
  $\vec{x}, \vec{y} \in \Z^{n}$ correspond bijectively to vectors
  $\vec{v} = (\vec{x}, \vec{y}) \in \Z^{n}$, and
  $\norm{\vec{v}}_{p} \leq \norm{\vec{v}'}_{p}$.

  For \cref{item:yes-num-close}, for every $\vec{y} \in \bit^{n-n'}$,
  the vector
  $\vec{v} := (\tfrac12 B' \vec{x}, \vec{x}, \vec{y}) \in \lat$
  satisfies
  \[ \norm{\vec{v} - \vec{t}}_p^{p} = \frac{{\norm{B'\vec{x} -
          \vec{t}'}_{p}^{p}}}{2^{p}} + \frac{n}{2^{p}} \leq
    s_{p}^{p}\] for finite~$p$, and
  $\norm{\vec{v}-\vec{t}}_{\infty} = \max(\tfrac12\norm{B'\vec{x} -
    \vec{t}'}_{\infty}, \tfrac12) = \tfrac12 = s_{\infty}$. The claim
  follows.

  For \cref{item:no-far}, for  finite~$p$ we have
  \[
    \dist_p(\vec{t}, \lat)^{p} \geq \frac{\dist_{p}(\vec{t}',
      \lat')^{p}}{2^{p}} + \frac{n}{2^{p}} > \frac{n+1}{2^{p}} =
    s_{p}^{p} \ , \] and for $p=\infty$ we immediately have
  $\dist_{\infty}(\vec{t}, \lat) \geq \tfrac12
  \dist_{\infty}(\vec{t}', \lat') > \tfrac12 = s_{\infty}$, as needed.
\end{proof}

\begin{corollary}
  \label{cor:cvp-to-st-bdd}
  For any $p \in [1,\infty]$, $\alpha > 0$, and $\poly(n')$-bounded
  $n \geq n'$, there is a deterministic Karp reduction from
  $\GapCVP'_{p}$ (resp., with preprocessing) in rank~$n'$ to the
  decision version of $(S, T)$-$\BDD_{p,\alpha}$ (resp., with
  preprocessing) in rank~$n$, where
  $S(n) = N^{o}_p(\Z^{n} \setminus \set{\vec{0}}, s_{p}/\alpha,
  \vec{0})$ for~$s_{p}$ as defined in \cref{eq:s_p}, and
  $T(n) = 2^{n-n'}$.
\end{corollary}

\begin{proof}
  Given an input $\GapCVP'_{p}$ instance $(B', \vec{t}')$, the
  reduction simply outputs $(B, r = s_{p}/\alpha, \vec{t})$, where
  $B, \vec{t}$ are as in \cref{eq:Bt}. Observe that this is also valid
  for the preprocessing problems because~$B$ and~$r$ depend only
  on~$B'$. Correctness follows immediately by
  \cref{lem:cvp-to-st-bdd-transform}.
\end{proof}

\subsection{Setting Parameters}
\label{sec:setting-params}

We now investigate the relationship among the choice of~$\ell_{p}$
norm (for finite~$p$), the BDD relative distance~$\alpha$, and the
rank ratio $C := n/n'$, subject to the constraint
\begin{equation}
  \label{eq:No-upper}
  N^{o}_p(\Z^{n}, s_{p}/\alpha, \vec{0}) \leq 2^{n-n'}/10 = T(n)/10 \ ,
\end{equation}
so that the reductions in \cref{cor:cvp-to-st-bdd,thm:st-bdd-to-bdd}
can be composed. For $p \in [1,\infty)$ and $C > 1$, define
\begin{align}
  \alpha_{p,C}^*
  &:= \inf \set{ \alpha^{*} > 0 : \min_{\tau > 0}
    \exp(\tau/(2\alpha^{*})^{p}) \cdot \Theta_p(\tau) \leq
    2^{1-1/C}} \ ,\label{eq:alpha_pC}\\
  \alpha_{p}^{*}
  &:= \lim_{C \to \infty} \alpha_{p,C}^{*}
    = \inf \set{ \alpha^{*} > 0 : \min_{\tau > 0}
    \exp(\tau/(2\alpha^{*})^{p}) \cdot \Theta_p(\tau) \leq 2}\ .\label{eq:alpha_p}
\end{align}
These quantities are well defined because for any $C > 1$ we have
$2^{1-1/C} > 1$, so the inequality in \cref{eq:alpha_pC} is satisfied
for sufficiently large~$\tau$ and~$\alpha^{*}$. Moreover, it is
straightforward to verify that~$\alpha_{p,C}^*$ is strictly decreasing
in both~$p$ and~$C$, and~$\alpha_{p}^{*}$ is strictly decreasing
in~$p$. Although it is not clear how to solve for these quantities in
closed form, it is possible to approximate them numerically to good
accuracy (see \cref{fig:bdd-hardness-bounds}), and to get quite tight
closed-form upper bounds (see~\cref{lem:analytic-ub-alpha-p-star}).
We now show that to satisfy \cref{eq:No-upper} it suffices to take any
constant $\alpha > \alpha_{p,C}^{*}$.

\begin{corollary}
  \label{cor:cvp-to-concrete-st-bdd}
  For any $p \in [1,\infty)$, $C \geq 1$, and constant
  $\alpha > \alpha_{p,C}^{*}$ (\cref{eq:alpha_pC}), there is a
  deterministic Karp reduction from $\GapCVP'_{p}$ (resp., with
  preprocessing) in rank $n'$ to the decision version of
  $(S, T)$-$\BDD_{p,\alpha}$ (resp., with preprocessing) in rank
  $n=Cn'$, where $S(n) = T(n)/10$ and $T(n) = 2^{(1-1/C)n}$.
\end{corollary}

\begin{proof}
  Recalling that $s_{p} = \tfrac12 (n+1)^{1/p}$, by \cref{prop:short-vec-ub},
  $N^{o}_{p}(\Z^{n}, s_{p}/\alpha, \vec{0})$ is at most
  \begin{align*}
    N_p(\Z^{n}, s_{p}/\alpha, \vec{0})
    &\leq \min_{\tau > 0} \exp(\tau \cdot (s_{p}/\alpha)^p) \cdot
      \Theta_p(\tau)^{n} \\
    &= \min_{\tau > 0} \exp(\tau \cdot (n+1)/(2\alpha)^{p}) \cdot
      \Theta_p(\tau)^{n} \\
    &= \parens[\Big]{\min_{\tau > 0} \exp(\tau / (n (2\alpha)^{p}))
      \cdot \exp(\tau/(2\alpha)^{p}) \cdot \Theta_{p}(\tau)}^{n} \ .
  \end{align*}
  Because $\alpha > \alpha_{p,C}^{*}$, we have that
  $\min_{\tau > 0} \exp(\tau/(2\alpha)^{p}) \cdot \Theta_{p}(\tau)$ is
  a constant strictly less than $2^{1-1/C}$. So,
  $N^{o}_{p}(\Z^{n}, s_{p}/\alpha, \vec{0}) \leq 2^{(1-1/C)n}/10 =
  T(n)/10$ for all large enough~$n$. The claim follows from
  \cref{cor:cvp-to-st-bdd}.
\end{proof}

\begin{theorem}
  \label{thm:cvp-to-bdd}
  For any $p \in [1,\infty)$, $C \geq 1$, and constant
  $\alpha > \alpha_{p,C}^{*}$, there is a randomized Cook reduction
  with no false positives from $\GapCVP'_{p}$ (resp., with
  preprocessing) in rank~$n'$ to $\BDD_{p,\alpha}$ (resp., with
  preprocessing) in rank~$n=Cn'+1$.  Furthermore, the same holds for
  $p=\infty$, $C=1$, $\alpha=1/2$, and the reduction is deterministic.
\end{theorem}

\begin{proof}
  For finite $p$, we simply compose the reductions from
  \cref{cor:cvp-to-concrete-st-bdd,thm:st-bdd-to-bdd}, with the
  trivial decision-to-search reduction for $(S,T)$-$\BDD_{p,\alpha}$
  in between.

  For $p = \infty$, we first invoke the deterministic reduction from
  \cref{cor:cvp-to-st-bdd}, from $\GapCVP'_{\infty}$ in rank~$n'$ to
  $(S,T)$-$\BDD_{\infty,1/2}$ in rank~$Cn' = n'$, where
  $S = N^{o}_{\infty}(\Z^{n} \setminus \set{\vec{0}}, 1, \vec{0}) = 0$
  and $T = 2^{0} > 0$. By \cref{lem:0t-bdd}, the latter problem
  reduces deterministically to $\BDD_{\infty,1/2}$ in rank
  $n'+1$.

  Lastly, all of these reductions work for the preprocessing problems
  as well, because their component reductions do.
\end{proof}

\subsection{Putting it all Together}
\label{sec:putting-together}

We now combine our reductions from $\GapSVP'$ to $\BDD$ with prior
hardness results for $\GapCVP'$ (stated below in
\cref{thm:hardness-of-cvp}) to obtain our ultimate hardness theorems
for $\BDD$. We first recall relevant known hardness results for
$\GapCVP'_p$ and $\GapCVPP'_p$.

\begin{theorem}[{\cite{journals/tit/Micciancio01,conf/focs/BennettGS17,aggarwal2019finegrained}}]
  \label{thm:hardness-of-cvp}
  The following hold for $\GapCVP'_p$ and $\GapCVPP'_p$ in rank~$n$:
  \begin{enumerate}[itemsep=0pt]
  \item For every $p \in [1, \infty]$, $\GapCVP'_p$ is $\NP$-hard, and
    $\GapCVPP'_p$ has no polynomial-time (preprocessing) algorithm
    unless $\NP \subseteq \PPoly$.\label{en:np-hardness}
  \item For every $p \in [1, \infty]$, there is no $2^{o(n)}$-time
    randomized algorithm for $\GapCVP'_p$ unless randomized ETH
    fails.\label{en:cvp-eth-hardness}
  \item For every $p \in [1, \infty] \setminus \set{2}$, there is no
    $2^{o(n)}$-time algorithm for $\GapCVPP'_p$, and there is no
    $2^{o(\sqrt{n})}$-time algorithm for $\GapCVPP_2'$, unless
    non-uniform ETH fails.\label{en:cvpp-eth-hardness}
  \item For every $p \in [1, \infty] \setminus 2\Z$ and every
    $\eps > 0$, there is no $2^{(1-\eps)n}$-time randomized algorithm
    for $\GapCVP'_p$ (respectively, $\GapCVPP'_p$) unless randomized
    SETH (resp., non-uniform SETH) fails.\label{en:seth-hardness}
  \end{enumerate}
\end{theorem}

\begin{remark}
  \label{rem:gapcvpp-variants}
  Several of the above results are stated slightly differently from
  what appears
  in~\cite{journals/tit/Micciancio01,conf/focs/BennettGS17,aggarwal2019finegrained}. First,
  all of the above results for $\GapCVP'_p$ (respectively,
  $\GapCVPP'_p$) are instead stated for $\GapCVP_p$ (resp.,
  $\GapCVPP_p$). However, inspection shows that the reductions are
  indeed to $\GapCVP'_p$ or $\GapCVPP'_p$, so this difference is
  immaterial.
	
  Second, the above statements ruling out \emph{randomized} algorithms
  for $\GapCVP'_p$ assuming randomized (S)ETH are instead phrased
  in~\cite{conf/focs/BennettGS17,aggarwal2019finegrained} as ruling
  out \emph{deterministic} algorithms for $\GapCVP'_p$ assuming
  deterministic (S)ETH. However, because these results are proved via
  deterministic reductions, randomized algorithms for $\GapCVP'_{p}$
  have the consequences claimed above.
	
  Third, the above results for $\GapCVPP_p'$ follow from the
  reductions given in (the proofs of)
  \cite{journals/tit/Micciancio01},
  \cite[Theorem~4.3]{aggarwal2019finegrained}, \cite[Theorem~1.4 and
  Lemma~6.1]{conf/focs/BennettGS17}, and
  \cite[Theorem~4.6]{aggarwal2019finegrained}.
  However, those reductions all prove hardness for the variant of
  $\GapCVPP'_p$ where the distance threshold~$r$ is part of the
  \emph{query} input, rather than the preprocessing input.  Inspection
  of~\cite[Theorem~4.6]{aggarwal2019finegrained} shows that~$r$ is
  fixed in the output instance, so this difference is immaterial in
  that case.  We next describe how to handle this difference for the
  remaining cases. Below we give, for any $p \in [1, \infty)$, a
  straightforward rank-preserving mapping reduction (in the sense of
  \cref{def:preprocessing-mapping-reduction}) from the variant of
  $\GapCVPP'_p$ where the distance threshold~$r$ is part of the query
  input to the variant where it is part of the preprocessing input,
  assuming that~$r$ is always at most some~$r^{*}$ that depends only
  on~$B$, and whose length $\log r^*$ is polynomial in the length
  of~$B$. Inspection shows that such an~$r^*$ does indeed exist for
  the reductions given in \cite{journals/tit/Micciancio01},
  \cite[Theorem~4.3]{aggarwal2019finegrained}, and
  \cite[Lemma~6.1]{conf/focs/BennettGS17}, which handles the second
  difference for those cases.

  The mapping reduction $(R_P, R_Q)$ in question maps
  $(B, (\vec{t}, r)) \mapsto ((B', r^{*}), \vec{t}')$ as
  follows. First, $R_P$ takes~$B$ as input, and sets
  $B' := \parens*{\begin{smallmatrix} B \\
      \vec{0}^t \end{smallmatrix}}$; it also outputs $\sigma' = r^{*}$
  as side information for $R_Q$. Then, $R_Q$ takes $(\vec{t},r)$ and
  $r^{*}$ as input, and outputs
  $\vec{t}' := (\vec{t}, ((r^{*})^p - r^p)^{1/p})$. Using the
  guarantee that $r^* \geq r$, it is straightforward to check that the
  output instance $((B', r^{*}), \vec{t}')$ is a YES instance
  (respectively, NO instance) if the input instance
  $(B, (\vec{t}, r))$ is a YES instance resp., NO instance, as
  required.
	
  Finally, we again remark that several of the hardness results
  in~\cref{thm:hardness-of-cvp} in fact hold under weaker versions of
  randomized or non-uniform (S)ETH that relate to Max-$3$-SAT
  (respectively, Max-$k$-SAT), instead of $3$-SAT
  (resp. $k$-SAT). Therefore, it is straightforward to obtain
  corresponding hardness results for BDD(P) under these weaker
  assumptions as well.
\end{remark}

\noindent We can now prove our main theorem, restated from the
introduction: \main*

\begin{proof}
  For BDD, each item of the theorem follows from the corresponding
  item of \cref{thm:hardness-of-cvp}, followed by
  \cref{thm:cvp-to-bdd} and then (where needed) rank-preserving norm
  embeddings from~$\ell_{2}$
  to~$\ell_{p}$~\cite{conf/stoc/RegevR06}. (Also,
  \cref{lem:analytic-ub-alpha-p-star} below provides the upper bound
  on~$\alpha_{p}^{*}$.) The claims for BDDP follow similarly, combined
  with the well-known fact that $\PPoly = \BPP/\mathsf{Poly}$
  and~\cref{cor:preprocessing-k-sat-rerandomization}.\footnote{In
    fact, $\PPoly = \BPP/\mathsf{Poly}$ also follows as a corollary of
    the more general derandomization result
    in~\cref{lem:weak-strong-equivalent}.}

\end{proof}

\subsection{An Upper Bound on \texorpdfstring{$\alpha_{p, C}^*$}{alpha\_(p,C){\textasciicircum}*} and \texorpdfstring{$\alpha_p^*$}{alpha\_p{\textasciicircum}*}}

We conclude with closed-form upper bounds on $\alpha_{p, C}^*$ and
$\alpha_p^*$. The main idea is to replace $\Theta_p(\tau)$ with an
upper bound of $\Theta_1(\tau)$ (which has a closed-form expression)
in \cref{eq:alpha_pC,eq:alpha_p}, then directly analyze the value of
$\tau > 0$ that minimizes the resulting expressions. This leads to quite
tight bounds (and also yields tighter bounds than the techniques used
in the proof of~\cite[Claim~4.4]{conf/stoc/AggarwalS18}, which bounds
a related quantity). For example, $\alpha_2^* \approx 1.05006$, and
the upper bound in~\cref{lem:analytic-ub-alpha-p-star} gives
$\alpha_2^* \leq 1.08078$; similarly, $\alpha_5^* \approx 0.672558$
and the upper bound in~\cref{lem:analytic-ub-alpha-p-star} gives
$\alpha_5^* \leq 0.680575$.

\begin{lemma}
  \label{lem:analytic-ub-alpha-p-star}
  Define
  \[ g(\sigma, \tau) := \exp(\tau/\sigma) \cdot \parens*{\frac{2}{1 -
        \exp(-\tau)} - 1} \] and
  $\tau^*(\sigma) := \arcsinh(\sigma) = \ln(\sigma + \sqrt{1 +
    \sigma^2})$.  Let $\sigma^*$ and $\sigma_C^*$ for $C > 1$ be the
  (unique) constants for which $g(\sigma^{*}, \tau^*(\sigma^{*})) = 2$
  and $g(\sigma_{C}^{*}, \tau^*(\sigma_{C}^{*})) = 2^{1 - 1/C}$. Then
  for any $p \in [1,\infty)$, we have
  \[
    \alpha_{p, C}^* \leq \frac{1}{2} \cdot (\sigma_C^*)^{1/p} \quad
    \text{and} \quad \alpha_p^* \leq \frac{1}{2} \cdot
    (\sigma^*)^{1/p} \leq \frac{1}{2} \cdot 4.6723^{1/p} \ .
  \]
  In particular, $\alpha_{p, C}^* \to 1/2$ as $p \to \infty$ for any
  fixed $C > 1$, and therefore $\alpha_p^* \to 1/2$ as $p \to \infty$.
\end{lemma}

\begin{proof}
  For any $\tau > 0$, by the definition of $\Theta_p(\tau)$ and the
  formula for summing geometric series we have
  \begin{equation}
    \Theta_p(\tau) \leq \Theta_1(\tau) =
    1 + 2 \sum_{i=1}^{\infty} \exp(-\tau)^i
    = \frac{2}{1 - \exp(-\tau)} - 1 \ .
    \label{eq:gs-ub-theta}
  \end{equation}
	
  Define the objective function
  \begin{align*}
    f(p, \alpha) &:= \min_{\tau > 0} \exp(\tau/(2\alpha)^p) \cdot \Theta_p(\tau)
  \end{align*}
  to be the expression that is upper-bounded in
  \cref{eq:alpha_pC,eq:alpha_p}.  For any fixed $\alpha > 0$, set
  $\sigma := (2 \alpha)^p$.  Applying \cref{eq:gs-ub-theta}, it
  follows that $f(p, \alpha) \leq g(\sigma, \tau)$ for any $\tau >
  0$. This implies that if there exists some $\tau > 0$ satisfying
  $g(\sigma, \tau) \leq 2$ then
  $\alpha_p^* \leq \frac{1}{2} \sigma^{1/p}$, and similarly, if
  $g(\sigma, \tau) \leq 2^{1 - 1/C}$ then
  $\alpha_{p,C}^* \leq \frac{1}{2} \sigma^{1/p}$.
	
  By standard calculus,
  \[
    \frac{\partial g}{\partial \tau} = \frac{e^{\tau/\sigma}}{1 -
      e^{-\tau}} \cdot \parens[\Big]{(1 + e^{-\tau})/\sigma - 2
      e^{-\tau}/(1 - e^{-\tau})} \ .
  \]
  Setting the right-hand side of the above expression equal to $0$ and
  solving for $\tau$ yields the single real solution
  \[
    \tau = \tau^*(\sigma) = \arcsinh(\sigma) = \ln(\sigma + \sqrt{1 + \sigma^2}) \ ,
  \]
  which is a local minimum, and therefore a global minimum of
  $g(\sigma, \tau)$ for any fixed $\sigma > 0$.
		
  Define the univariate function
  $g^{*}(\sigma) := g(\sigma, \tau^*(\sigma))$. The fact that
  $\sigma^*$ and $\sigma_C^*$ exist and are unique follows by noting
  that $\lim_{\sigma \to 0^+} g^{*}(\sigma) = \infty$, that
  $\lim_{\sigma \to \infty} g^{*}(\sigma) = 1$, and that
  $g^{*}(\sigma)$ is strictly decreasing in $\sigma > 0$. By
  definition of $\sigma^*$ (respectively, $\sigma_C^*$), it follows
  that $g^{*}(\sigma^{*}) = 2$ for
  $\alpha = \frac{1}{2} (\sigma^*)^{1/p}$, and
  $g^{*}(\sigma^{*}_{C}) = 2^{1 - 1/C}$ for
  $\alpha = \frac{1}{2} (\sigma_C^*)^{1/p}$, as desired.  Moreover,
  one can check numerically that $\sigma^* \leq 4.6723$.
\end{proof}

\bibliography{bdd}

\begin{thebibliography}{ADRS15}
\expandafter\ifx\csname urlstyle\endcsname\relax
  \providecommand{\doi}[1]{doi:\discretionary{}{}{}#1}\else
  \providecommand{\doi}{doi:\discretionary{}{}{}\begingroup
  \urlstyle{rm}\Url}\fi

\bibitem[ABGS19]{aggarwal2019finegrained}
D.~Aggarwal, H.~Bennett, A.~Golovnev, and N.~{Stephens-Davidowitz}.
\newblock {Fine-grained hardness of CVP(P)--- Everything that we can prove (and
  nothing else)}.
\newblock \url{https://arxiv.org/abs/1911.02440}, 2019.

\bibitem[ABSS97]{journals/jcss/AroraBSS97}
S.~Arora, L.~Babai, J.~Stern, and Z.~Sweedyk.
\newblock The hardness of approximate optima in lattices, codes, and systems of
  linear equations.
\newblock \emph{J. Comput. Syst. Sci.}, 54(2):317--331, 1997.

\bibitem[Adl78]{DBLP:conf/focs/Adleman78}
L.~M. Adleman.
\newblock Two theorems on random polynomial time.
\newblock In \emph{FOCS}, pages 75--83. 1978.

\bibitem[ADRS15]{conf/stoc/AggarwalDRS15}
D.~Aggarwal, D.~Dadush, O.~Regev, and N.~Stephens{-}Davidowitz.
\newblock Solving the shortest vector problem in $2^n$ time using discrete
  {Gaussian} sampling: Extended abstract.
\newblock In \emph{STOC}, pages 733--742. 2015.

\bibitem[ADS15]{conf/focs/AggarwalDS15}
D.~Aggarwal, D.~Dadush, and N.~Stephens{-}Davidowitz.
\newblock Solving the closest vector problem in $2^n$ time -- the discrete
  {Gaussian} strikes again!
\newblock In \emph{FOCS}, pages 563--582. 2015.

\bibitem[Ajt98]{conf/stoc/Ajtai98}
M.~Ajtai.
\newblock The shortest vector problem in {$L_{2}$} is {NP}-hard for randomized
  reductions (extended abstract).
\newblock In \emph{STOC}, pages 10--19. 1998.

\bibitem[AS18a]{conf/stoc/AggarwalS18}
D.~Aggarwal and N.~{Stephens-Davidowitz}.
\newblock {(Gap/S)ETH} hardness of {SVP}.
\newblock In \emph{STOC}, pages 228--238. 2018.

\bibitem[AS18b]{conf/soda/AggarwalS18}
D.~Aggarwal and N.~Stephens{-}Davidowitz.
\newblock Just take the average! {An} embarrassingly simple $2^n$-time
  algorithm for {SVP} (and {CVP)}.
\newblock In \emph{Symposium on Simplicity in Algorithms}, volume~61, pages
  12:1--12:19. 2018.

\bibitem[BGS17]{conf/focs/BennettGS17}
H.~Bennett, A.~Golovnev, and N.~Stephens{-}Davidowitz.
\newblock On the quantitative hardness of {CVP}.
\newblock In \emph{FOCS}. 2017.

\bibitem[BSW16]{conf/icalp/BaiSW16}
S.~Bai, D.~Stehl{\'{e}}, and W.~Wen.
\newblock Improved reduction from the bounded distance decoding problem to the
  unique shortest vector problem in lattices.
\newblock In \emph{ICALP}, pages 76:1--76:12. 2016.

\bibitem[DRS14]{conf/coco/DadushRS14}
D.~Dadush, O.~Regev, and N.~Stephens{-}Davidowitz.
\newblock On the closest vector problem with a distance guarantee.
\newblock In \emph{IEEE Conference on Computational Complexity}, pages 98--109.
  2014.

\bibitem[EOR91]{elkies91:_packing_densities}
N.~D. Elkies, A.~M. Odlyzko, and J.~A. Rush.
\newblock On the packing densities of superballs and other bodies.
\newblock \emph{Inventiones mathematicae}, 105:613–639, December 1991.

\bibitem[GPV08]{DBLP:conf/stoc/GentryPV08}
C.~Gentry, C.~Peikert, and V.~Vaikuntanathan.
\newblock Trapdoors for hard lattices and new cryptographic constructions.
\newblock In \emph{STOC}, pages 197--206. 2008.

\bibitem[HR07]{journals/toc/HavivR12}
I.~Haviv and O.~Regev.
\newblock Tensor-based hardness of the shortest vector problem to within almost
  polynomial factors.
\newblock \emph{Theory of Computing}, 8(1):513--531, 2012.
\newblock Preliminary version in STOC 2007.

\bibitem[IP01]{journals/jcss/ImpagliazzoP01}
R.~Impagliazzo and R.~Paturi.
\newblock On the complexity of $k$-{SAT}.
\newblock \emph{J. Comput. Syst. Sci.}, 62(2):367--375, 2001.

\bibitem[Kho03]{journals/jcss/Khot06}
S.~Khot.
\newblock Hardness of approximating the shortest vector problem in high
  $\ell_p$ norms.
\newblock \emph{J. Comput. Syst. Sci.}, 72(2):206--219, 2006.

\bibitem[Kho04]{journals/jacm/Khot05}
S.~Khot.
\newblock Hardness of approximating the shortest vector problem in lattices.
\newblock \emph{J.~ACM}, 52(5):789--808, 2005.
\newblock Preliminary version in FOCS 2004.

\bibitem[KS01]{journals/tcs/KumarS01}
R.~Kumar and D.~Sivakumar.
\newblock On the unique shortest lattice vector problem.
\newblock \emph{Theor. Comput. Sci.}, 255(1-2):641--648, 2001.

\bibitem[LLM06]{conf/approx/LiuLM06}
Y.~Liu, V.~Lyubashevsky, and D.~Micciancio.
\newblock On bounded distance decoding for general lattices.
\newblock In \emph{APPROX-RANDOM}, pages 450--461. 2006.

\bibitem[LM09]{conf/crypto/LyubashevskyM09}
V.~Lyubashevsky and D.~Micciancio.
\newblock On bounded distance decoding, unique shortest vectors, and the
  minimum distance problem.
\newblock In \emph{CRYPTO}, pages 577--594. 2009.

\bibitem[Mic98]{journals/siamcomp/Micciancio00}
D.~Micciancio.
\newblock The shortest vector in a lattice is hard to approximate to within
  some constant.
\newblock \emph{SIAM J.~Comput.}, 30(6):2008--2035, 2000.
\newblock Preliminary version in FOCS 1998.

\bibitem[Mic01]{journals/tit/Micciancio01}
D.~Micciancio.
\newblock The hardness of the closest vector problem with preprocessing.
\newblock \emph{{IEEE} Trans. Information Theory}, 47(3):1212--1215, 2001.
\newblock \doi{10.1109/18.915688}.

\bibitem[Mic08]{conf/soda/Micciancio08}
D.~Micciancio.
\newblock Efficient reductions among lattice problems.
\newblock In \emph{SODA}, pages 84--93. 2008.

\bibitem[Mic12]{journals/toc/Micciancio12}
D.~Micciancio.
\newblock Inapproximability of the shortest vector problem: Toward a
  deterministic reduction.
\newblock \emph{Theory of Computing}, 8(1):487--512, 2012.

\bibitem[MO90]{mazo90:_lattice_points}
J.~E. Mazo and A.~M. Odlyzko.
\newblock Lattice points in high-dimensional spheres.
\newblock \emph{Monatshefte f\"{u}r Mathematik}, 110:47--61, March 1990.

\bibitem[Pei09]{conf/stoc/Peikert09}
C.~Peikert.
\newblock Public-key cryptosystems from the worst-case shortest vector problem.
\newblock In \emph{STOC}, pages 333--342. 2009.

\bibitem[Reg05]{journals/jacm/Regev09}
O.~Regev.
\newblock On lattices, learning with errors, random linear codes, and
  cryptography.
\newblock \emph{J.~ACM}, 56(6):1--40, 2009.
\newblock Preliminary version in STOC 2005.

\bibitem[RR06]{conf/stoc/RegevR06}
O.~Regev and R.~Rosen.
\newblock Lattice problems and norm embeddings.
\newblock In \emph{STOC}, pages 447--456. 2006.

\bibitem[Ste16]{conf/soda/Stephens-Davidowitz16}
N.~Stephens{-}Davidowitz.
\newblock Discrete {Gaussian} sampling reduces to {CVP} and {SVP}.
\newblock In \emph{SODA}, pages 1748--1764. 2016.

\bibitem[SV19]{conf/focs/Stephens-Davidowitz19}
N.~Stephens{-}Davidowitz and V.~Vaikuntanathan.
\newblock {SETH}-hardness of coding problems.
\newblock In \emph{FOCS}, pages 287--301. 2019.

\bibitem[vEB81]{emde81:_anoth_np}
P.~van Emde~Boas.
\newblock Another {NP}-complete problem and the complexity of computing short
  vectors in a lattice.
\newblock Technical Report 81-04, University of Amsterdam, 1981.

\end{thebibliography}
\bibliographystyle{alphaabbrvprelim}

\end{document}